\documentclass[12pt]{iopart}
 
\usepackage{iopams}

\expandafter\let\csname equation*\endcsname\relax

\expandafter\let\csname endequation*\endcsname\relax

\usepackage{amsmath}

\usepackage{graphicx}  
\usepackage{dcolumn}   
\usepackage{bm}        
\usepackage{bbm}       
\usepackage{amssymb}   
\usepackage{tabularx}
\usepackage{amsfonts}
\usepackage{amsthm}
\usepackage{txfonts}
\usepackage{arydshln}
\usepackage{multirow}

\usepackage{color}

\newtheorem*{proposition}{Proposition}

\newtheorem{definition}{Definition}
\newtheorem*{theorem-non}{Theorem}
\newtheorem*{corollary-non}{Corollary}
\newtheorem*{summary-non}{Summary}
\usepackage{bm}

\def\ket#1{\mathinner{|{#1}\rangle}}

\def\Id{\operatorname{Id}}

\newcommand{\egale}[1]{%
\ensuremath{\stackrel{#1}{=}}}
\newcommand{\eegale}[2]{%
\ensuremath{\underset{#2}{\overset{#1}{=}}}}

\makeatletter
\newcommand{\myLines}[1]{
\begin{picture}(4,1)
\put(0,0){\line(2,1){2}}
\put(2,0){\vector(0,1){0.7}}
\put(4,0){\line(-2,1){2}}
\end{picture}}
\let\qdrawReal=\qdraw@branches
\newcommand\brOverride{\let\qdraw@branches=\myLines}
\newcommand\brRestore{\let\qdraw@branches=\qdrawReal}
\makeatother

\newcommand{\coh}[2]{\mathinner{|{#1}\rangle}\!\!\mathinner{\langle{#2}|}}

\def\text#1{\textrm{#1}}
\def\id{\mathbbm{1}}
\def\eq{\begin{equation}}
\def\eeq{\end{equation}}

\newcommand{\theoremname}{{\bf{Theorem}}}

\newcommand{\sectionname}{Sec.~}

\begin{document}

\title[On the inequivalence of the CH and CHSH inequalities due to finite statistics]{On the inequivalence of the CH and CHSH inequalities due to finite statistics}


\author{M O Renou$^1$, D Rosset$^{1,2}$, A Martin$^1$ and N Gisin$^1$}
\address{1. Group of Applied Physics, University of Geneva, Switzerland\\2. National Cheng Kung University, Tainan, Taiwan}
\ead{marcolivier.renou@unige.ch}

\begin{abstract}
Different variants of a Bell inequality, such as CHSH and CH, are known to be equivalent when evaluated on nonsignaling outcome probability distributions. However, in experimental setups, the outcome probability distributions are estimated using a finite number of samples. Therefore the nonsignaling conditions are only approximately satisfied and the robustness of the violation depends on the chosen inequality variant. We explain that phenomenon using the decomposition of the space of outcome probability distributions under the action of the symmetry group of the scenario, and propose a method to optimize the statistical robustness of a Bell inequality. In the process, we describe the finite group composed of relabeling of parties, measurement settings and outcomes, and identify correspondences between the irreducible representations of this group and properties of outcome probability distributions such as normalization, signaling or having uniform marginals.
\end{abstract}

\pacs{03.65.Ud, 02.20.-a}

\noindent{\it Keywords\/}: Bell inequalities, Nonlocality, Device-independent quantum information, Symmetries.

\maketitle

\section{Introduction}\label{sec:intro}

The structure of physical systems can be probed by observing relations between the measurement outcomes of their subsystems. When the subsystems are spacelike separated, special relativity forbids faster-than-light communication; then the outcome probability distributions obey the linear nonsignaling constraints~{\cite{Cirelson1980,Popescu1994,Brunner2014}}. Additionally, when subsystems obey the locality principle, the outcome probability distributions respect linear constraints known as Bell inequalities~{\cite{Bell1964}}.  The CHSH~{\cite{Clauser1969}} and CH~{\cite{Clauser1974}} Bell inequalities were introduced by Clauser et al., along with experimental proposals to test the violation of the locality principle by quantum mechanics. Later on, the CHSH and the CH inequalities were shown to be equivalent~{\cite{Mermin1995,Cereceda2001}} provided the outcome probability distributions satisfy the nonsignaling constraints.

\begin{figure}
\includegraphics{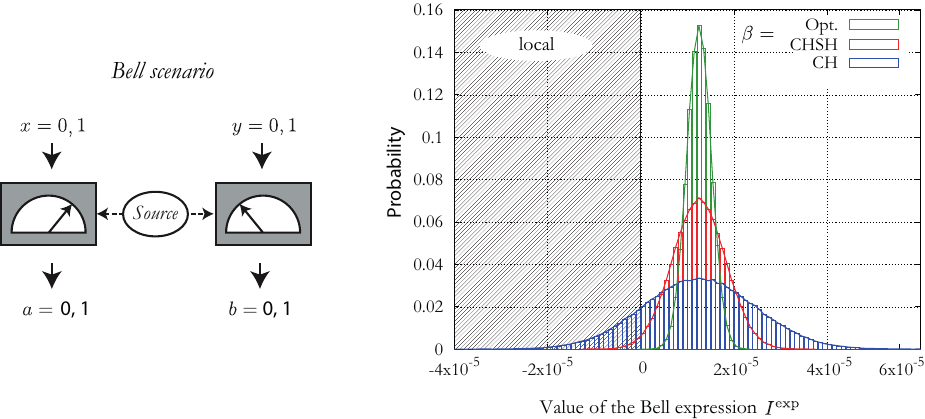}
\caption{\label{fig:example_intro} Simulated tests of variants of Bell inequalities modeled after the 2015 loophole-free experiment of Shalm et al.~\cite{Shalm2015}. The behavior of the system is represented by the distribution $p_{ab|xy}$, estimated using $P^\mathrm{exp}(ab|xy) = N(abxy) / N(xy)$. Due to finite sample effects, $\vec{P}^\mathrm{exp}$ violates the nonsignaling constraints, rendering CH and CHSH inequivalent (the inequalities have been rescaled to allow comparison, and Opt is the optimal variant obtained in Section~\ref{sec:optimal}). The histogram represents 200'000 runs of the experiment, each run including the same number of trials as in~\cite{Shalm2015}; while the average value of the Bell expression $I^\mathrm{exp}$ is the same, the variance differs significantly.
}
\end{figure}

Quantum systems respect the nonsignaling constraints; thus, in theory, the CHSH and CH inequalities should be equivalent. However, in practice, experimental probability distributions are estimated using a finite set of samples, and the nonsignaling constraints are only satisfied approximately due to statistical fluctuations. Thus, supposedly equivalent Bell inequalities exhibit differing statistical properties, as hinted by Knill et al.~\cite{Knill2015} and Gill~\cite{Gill2015}. In particular, the standard deviation of the reported Bell inequality violation is affected, as shown in Figure~\ref{fig:example_intro} for the CHSH and CH inequalities. With this observation, we can look for the optimal variant of a given Bell inequality. To do so, we decompose the inequality into a nonsignaling subspace (left untouched) and a signaling subspace to be optimized. While several ways to decompose a Bell inequality have already been proposed~\cite{Brunner2014,Collins2004,Rosset2014a}, we show the existence of a unique decomposition that satisfies the requirements of the device-independent framework~~{\cite{Acin2007,Bardyn2009,Pironio2009}.

Expanding on our earlier investigations~\cite{Rosset2014a}, our method rests on the study of the symmetries of Bell scenarios. In the device-independent approach, nonlocal behaviors are studied without attributing a particular meaning to the parties, the measurement settings and the outcomes, which are, for all purposes, abstract labels --- no assumptions are made about the inner workings of devices. Likewise, the set of Bell inequalities is defined up to the relabeling of parties, settings and outcomes~\cite{Avis2005,Avis2006}. Quite naturally, the symmetry group of those relabelings collects Bell inequalities into families~{\cite{Pitowsky1991,Fine1982,Werner2001,Masanes2003,Sliwa2003,Rosset2014a}}. However, in the current work, we assume that we already selected the representative of a Bell inequality to be violated in an experimental setup, and merely optimize its signaling subspace.

The required decomposition of the inequality coefficients into subspaces follows straightforwardly from the structure of the action of the relabeling group: for example, a nonsignaling probability distribution stays nonsignaling after an arbitrary relabeling. The same phenomenon holds for other properties, such as being properly normalized or having uniformly random marginals. These three properties correspond to invariant subspaces of the probability distribution coefficients; a complete list of these subspaces is given by the irreducible representations of the relabeling group.

The present paper has two goals. First, we present a practical method to optimize the Bell inequality variant used in an experiment. Second, we enumerate all the irreducible representations of the relabeling group present in outcome probability distributions and associate them to physical properties. Accordingly, we organize our paper as follows. In Section~\ref{sec:labelingequivprinciple}, we introduce the relabeling equivalence principle motivating our decomposition. We describe our optimization process in Section~\ref{sec:optimal} and illustrate it with simulations of two recent loophole-free Bell experiments. Together, these two sections form a self-contained description of our optimization method. In Section~\ref{sec:decompo}, we list all the invariant subspaces appearing in the outcome probability distributions of a bipartite scenario with binary settings and outcomes.

\subsection{Definitions used through out the present work}

A {\em Bell scenario} is described by the triple $(n,m,k)$, where $n$ is the number of parties, each with $m$ measurement settings and $k$ measurement outcomes per settings~\cite{AllTheBell}. We write $(x, y, ...)$ the measurement settings used by the parties, and $(a, b, ...)$ the obtained measurement outcomes. In a Bell scenario, a {\em setup} is a concrete or gedanken experiment, whose behavior is described by the joint conditional probability distribution $p_{ab...|xy...}$. Using a suitable enumeration of the coefficients, this distribution can be written as a vector $\vec{P} \in \mathbb{R}^d$, where $d = (m k)^n$.

When a quantum setup can be modeled precisely, the distribution $p_{ab...|xy...}$ is exactly given by the Born rule. In experimental situations, however, the distribution $\vec{P}$ is estimated by an experimental {\em run} recording the outcomes of a certain number of {\em trials} in the counts $N(ab...xy...)$. From these counts, the outcome probability distribution $\vec{P}$ is approximated using the relative frequencies $p_{ab...|xy...} \approx N(ab...xy...)/N(xy...)$. A recent example is given by the loophole-free Bell test of Shaml et al.~\cite{Shalm2015}~(Supplemental Material): in the scenario $(2,2,2)$, the authors present the event counts of six experimental runs, each run collecting approximately $2 \cdot 10^8$ trials.

\section{The label equivalence principle and representations}
\label{sec:labelingequivprinciple}

We first derive a general principle from the device-independent framework which motivates our introduction of elements of representation theory in the study of nonlocality. For simplicity, in this paper, we will restrict ourselves to the (2,2,2) scenario in which the CHSH inequality is tested. Our results generalizes straightforwardly~\cite{tobepublished}.
In that scenario, the behavior of a setup is described by $p_{ab|xy}$ with $a,b=0,1$ and $x,y=0^*,1^*$ (the asterisk is purely cosmetic to distinguish settings and outcomes). The same probability distribution is interchangeably written as the vector $\vec{P} \in V = \mathbb{R}^{d=16}$.

\begin{figure}
\centering
\includegraphics[width =.5 \columnwidth]{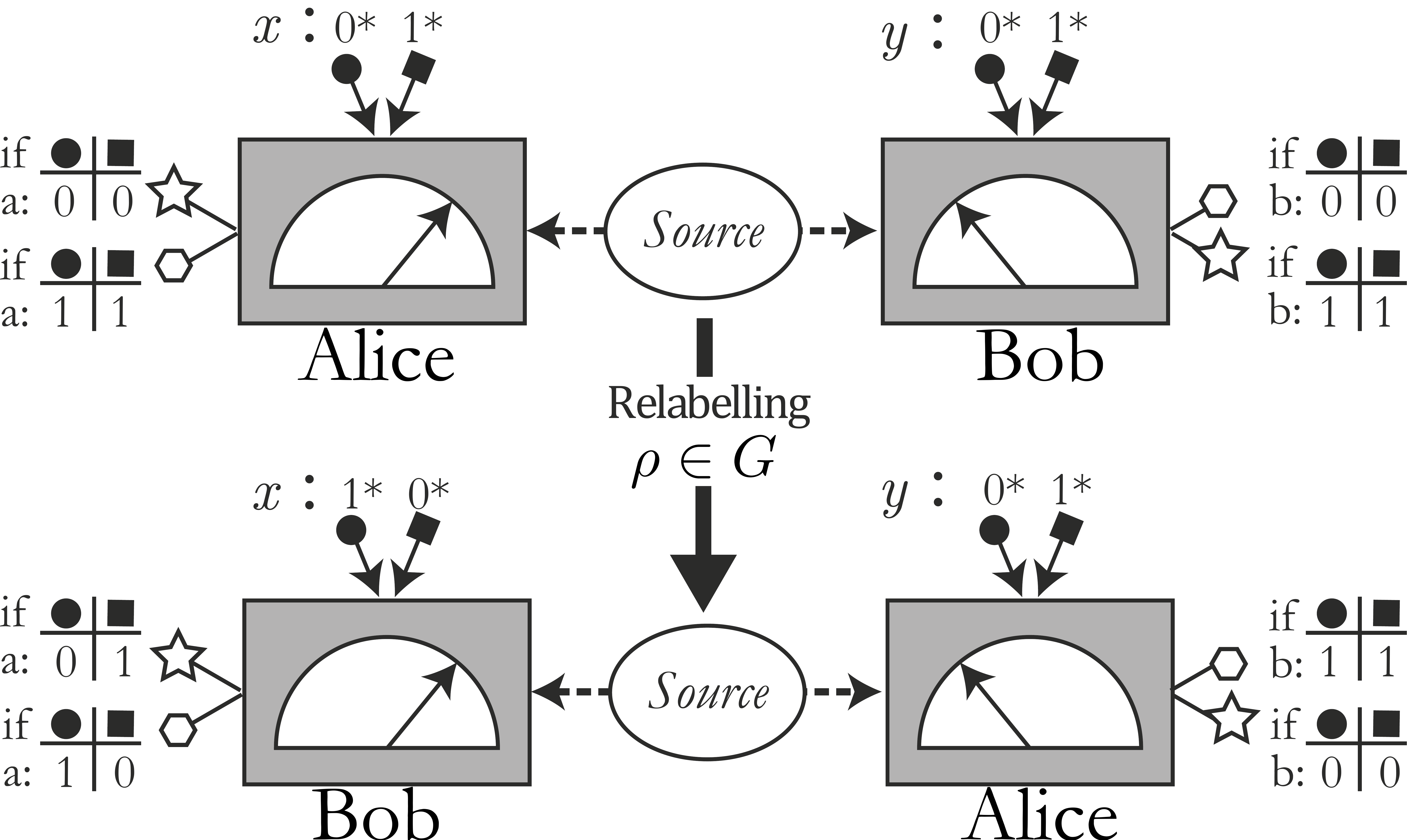}
\caption{\label{fig:relab}An experiment testing the CHSH inequality in the (2,2,2) scenario in a given labeling convention. Below, the same experiment using a different labeling convention, along with a relabeling map $\rho\in G$. Note that $\rho$ may exchange $a=0$ with $a=1$ when $x=0^*$ and do nothing when $x=1^*$.}
\end{figure}

However, to write the probability distribution $p_{ab|xy}$ above, a certain {\em labeling convention} has to be used, deciding who is Alice ($a,x$) and Bob ($b,y$), and which measurement setting (respectively measurement outcome) corresponds to $0^*$ and $1^*$ (respectively $0$, $1$). The transformation between labeling conventions is given by the {\em relabeling group} $G_{nmk}$, written $G$ when no confusion is possible. It is fully described in~\ref{app:step} and contains the relabelings of the parties, the settings, and the outcomes: a specific relabeling is given in~\figurename{~\ref{fig:relab}}. In a given scenario, the labels of parties, settings, and outcomes are purely abstract and have no physical meaning. This leads to the following principle:

\begin{definition}[Label equivalence principle]
In the study of nonlocality and device-independent protocols, all labeling conventions are equivalent.
\end{definition}

This principle is already present in all study of nonlocality. As an example, let us consider the constraints satisfied by nonsignaling probability distributions. All probability distributions are nonnegative:
\begin{equation}\label{eq:nonnegative}
\forall a,b,x,y, \quad p_{ab|xy} \ge 0,
\end{equation}
and normalized:
\begin{equation}\label{eq:normalization}
\forall x, y, \sum_{ab} p_{ab|xy} = 1.
\end{equation}
In addition, the nonsignaling probability distributions satisfy:
\begin{equation}\label{eq:nonsignaling}
\forall x, y, a, \sum_b p_{ab|xy} = p_{a|x} \text{ and } \forall x, y, b, \sum_a p_{ab|xy} = p_{b|y}.
\end{equation}
Taken as sets of constraints, each of Eqs.~\eqref{eq:nonnegative}-\eqref{eq:nonsignaling} satisfy the label equivalence principle.\\
To understand the deeper phenomenon at play here, let us make an analogy. In special relativity, the Lorentz group specifies the transformations between reference frames, and in turn, the physics of a relativistic system is contained in the representations of the Lorentz group. Similarly, a natural description of quantum behaviors should be given by the representations of the relabeling group $G$. Already, $V \ni \vec{P}$ is a representation of $G$, as the linear action of $G$ is to permute the coefficients of the probability distribution $p_{ab|xy}$. We then decompose $V$ into the invariant subspaces corresponding to irreducible representations of $G$~\cite{Serre}, and match the resulting subspaces to physical properties. Some of these properties are readily identified: under $G$, properly normalized distributions stay normalized and nonsignaling distributions stay nonsignaling. The full decomposition of $V$ into irreducible representations is presented in Section~\ref{sec:decompo}. To achieve our goal of optimizing the coefficients of a Bell inequality, a coarser-grained decomposition into three orthogonal subspaces is sufficient: the first subspace corresponds to the normalization conditions Eq.~\eqref{eq:normalization}, the second to nonsignaling content satisfying Eq.~\eqref{eq:nonsignaling}, the last to the signaling content of the probability distribution. This decomposition is the unique one respecting the label equivalence principle.\footnote{
Let us make a side remark. What is the most fundamental object in the study of nonlocality? On the one hand, we can argue that $p_{ab...|xy...}$ is the fundamental object. The composition of the relabeling group $G$ follows from the nonsignaling constraints: nonnegativity Eq.~\eqref{eq:nonnegative} prescribes a permutation representation on $\vec{P} \in V = \mathbb{R}^d$, so that $G$ is isomorphic to a subgroup of the symmetric group $\mathfrak{S}_d$. The exact subgroup is the one preserving the normalization Eq.~\eqref{eq:normalization} and nonsignaling Eq.~\eqref{eq:nonsignaling} constraints, recovering $G$. On the other hand, the relabeling group $G$ can be taken as the fundamental object, whose irreducible representations provide directly the components of $\vec{P}$ respecting the nonsignaling principle. However, there are other irreducible representations of $G$ that do not appear in the decomposition of $V \ni \vec{P}$. It is an open question whether they play a role in the study of nonlocality.
}

\section{Constructing the optimal variant of a Bell inequality}
\label{sec:optimal}

Having motivated our decomposition of $\vec{P} \in V$ into normalization, nonsignaling and signaling subspaces, we look to optimize the statistical properties of a given Bell inequality. We consider a bipartite nonsignaling setup. This setup can be either a classical or quantum experiment, or even a gedanken  experiment simulated on a computer. Because we study the nonlocal properties of the setup, we immediately forget about the description of the physical system and measurements to only remember the probabilities $p^{\,\rm setup}_{ab|xy}$ that define it. As the distribution $p^{\,\rm setup}_{ab|xy}$ is nonsignaling, it satisfies Eqs.~\eqref{eq:nonnegative}-\eqref{eq:nonsignaling}.

When the setup exhibits nonlocality, the probability distribution $p^{\,\rm setup}_{ab|xy}$  violates a Bell inequality. This Bell inequality can be written using coefficients $\beta_{abxy}$ along with a local bound $u$ such that
\begin{equation}\label{eq:bellinequalityvalue}
I = \sum_{abxy} \beta_{abxy} p_{ab|xy} = \vec{\beta} \cdot \vec{P} \overset{\mathrm{local}}{\le} u,
\end{equation}
is satisfied when the distribution $p_{ab|xy}$ is local. In the above equation, we also write $\beta_{abxy}$ as a vector $\vec{\beta} \in V = \mathbb{R}^d$ such that the computation of the inequality value is an inner product. Our method is based on the following proposition, which is a simplified version of the Theorem presented in Section~\ref{sec:decompo}:
\begin{proposition}
$V$ splits into three orthogonal subspaces $V = V_\mathrm{NO} \oplus V_\mathrm{NS} \oplus V_\mathrm{SI}$ invariant under $G$. Accordingly, the coefficient vectors of probability distributions $\vec{P}$ and Bell inequalities $\vec{\beta}$ can be decomposed into orthogonal vectors:
\begin{equation}
\vec{P} = \vec{P}_\mathrm{NO} + \vec{P}_\mathrm{NS} + \vec{P}_\mathrm{SI}, \quad 
\vec{\beta} = \vec{\beta}_\mathrm{NO} + \vec{\beta}_\mathrm{NS} + \vec{\beta}_\mathrm{SI}.
\end{equation}
In this decomposition, $\vec{P}_\mathrm{NO}$,  $\vec{P}_\mathrm{NS}$ and $\vec{P}_\mathrm{SI}$ correspond to the normalization, the nonsignaling content and the signaling content of the probability distribution $\vec{P}$, with the following relations:
\begin{enumerate}
\item $\vec{P}_\mathrm{NO}= \frac{1}{4} \vec{1}$ for all probability distributions, where $\vec{1}=(1,1,...,1)^\top$ (normalization),
\item $\vec{P}_\mathrm{SI} = \vec{0}$ for all nonsignaling probability distributions.
\end{enumerate}

\end{proposition}
\begin{proof}
See Section~\ref{sec:decompo} and~\ref{app:step}.
\end{proof}

\begin{figure}
\centering
\includegraphics[width =.5 \columnwidth]{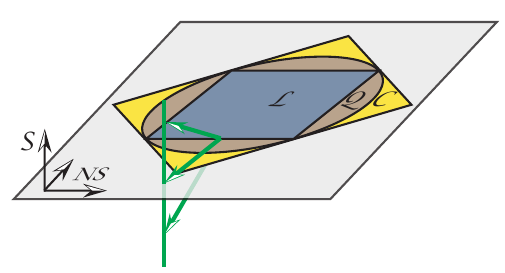}
\caption{\label{fig:decomp} Variants of Bell inequalities $\vec{\beta}$, which only differ by their component $\vec{\beta}_\mathrm{SI}$ over the signaling subspace. These variants are equivalent when evaluated on nonsignaling probability distributions. The letters $\mathcal{L}$, $\mathcal{Q}$, $\mathcal{C}$ represent respectively the sets of local, quantum, nonsignaling probability distributions~\cite{Brunner2014}.}
\end{figure}

This decomposition can also be deduced from Section 2. Indeed, the Eq.~\ref{eq:normalization} and \ref{eq:nonsignaling} are linear conditions on $\vec{P}$, which can be directly translated into equations about a projection of $\vec{P}$. Here Eq.~\ref{eq:normalization} and \ref{eq:nonsignaling} respectively correspond to (i), over projection $\vec{P}_{NO}$, and (ii), over projection $\vec{P}_{SI}$. The last component $\vec{P}_{NS}$, which is not constrained by normalization and nonsignaling (and then contains all the physical information), is defined as $\vec{P}_{NS}=\vec{P}-\vec{P}_{NO}-\vec{P}_{SI}$.

Using that proposition and the orthogonality relations, we rewrite the scalar product Eq.~\eqref{eq:bellinequalityvalue}~:
\begin{equation}
 I = 
\vec{\beta} \cdot \vec{P} = 
\underset{= \text{ cte}}{(\vec{\beta}_\mathrm{NO} \cdot \vec{P}_\mathrm{NO})} + 
(\vec{\beta}_\mathrm{NS} \cdot \vec{P}_\mathrm{NS}) + 
\underset{=0 \text{ if nonsignaling}}{(\vec{\beta}_\mathrm{SI} \cdot \vec{P}_\mathrm{SI})} \overset{\mathrm{local}}{\leq} u,
\end{equation}

enabling the following transformations of a Bell inequality:

\begin{itemize}
\item The local bound $u$ of the inequality can easily be translated by modifying the coefficients $\vec{\beta}_\text{NO}$, 
\item The inequality can be rescaled by multiplying both $\vec{\beta}$ and $u$ by the  same positive factor.
\item The coefficients $\vec{\beta}_\mathrm{SI}$ can be set to an arbitrary value.
\end{itemize}

We see from the first two transformations that a comparison of the violation $I$ across Bell inequality variants is only possible if the variants are first rescaled such that each of their local and quantum bounds match. Such a rescaling has no physical meaning and is pure convention.  What happens with the third transformation is more surprising\footnote{Some papers report the ratio $s = (I_\mathrm{exp} - u)/\sigma$ representing the magnitude of the violation as the number of standard deviations above the local bound. This ratio $s$ is invariant under translation and scaling.}. In experiments, the standard deviation $\sigma$ of the violation $I^\mathrm{exp}=\vec{\beta} \cdot \vec{P}^\mathrm{exp}$ has a dependence on $\vec{\beta}_\mathrm{SI}$, and thus the component $\vec{\beta}_\mathrm{SI}$ can be optimized.

\subsection{Spurious signaling due to finite statistics}\label{sec:spusignfinstat}
In an experimental {\em run}, the distribution $p^{\,\rm setup}_{ab|xy}$ is estimated by repeating measurements for a certain number $N$ of trials. For each trial, a setting pair $(x,y)$ is selected and the outcomes $(a,b)$ of the corresponding measurements are recorded. Writing $N(abxy)$ the number of events corresponding to settings $(x,y)$ and outcomes $(a,b)$, we estimate:
\begin{equation}\label{normalization}
p^{\,\rm run}_{ab|xy} = \frac{N(abxy)}{N(xy)}.
\end{equation}
We saw that $\vec{P}^{\,\rm setup}_\mathrm{SI} = \vec{0}$. However, because $\vec{P}^{\,\rm run}$ is estimated using a finite number of samples, statistical fluctuations lead to signaling outcome probability distributions: for a particular run, we generally have $\vec{P}^{\,\rm run}_\mathrm{SI} \ne \vec{0}$. Still, by construction, the normalization component $\vec{P}^{\,\rm run}_\mathrm{NO}$ is always constant. 
 By the central limit theorem, $\vec{P}^{\,\rm run}$ is well approximated by a multivariate normal distribution $\mathcal{N}(\vec{\mu}, \Sigma)$, where the mean vector $\vec{\mu} = \vec{P}^{\,\rm setup}$ and the covariance matrix $\Sigma_{abxy,a'b'x'y'}$ could be computed from the multinomial distribution. In practice, the Monte-Carlo estimation of $\Sigma$ performs well.

For each run, we can compute the Bell violation $I^{\,\rm run} = \vec{\beta} \cdot \vec{P}^{\,\rm run}$. If we were to perform infinitely many trials in each run, the average value $\left < I^{\,\rm run} \right >$ of the Bell inequality would be given by $I^{\,\rm setup}$, as the average signaling component $\left < \vec{P}^{\,\rm run}_\mathrm{SI} \right >$ would be zero. However, a run is composed of a finite number of trials, creating statistical fluctuations. The variance $\sigma$ of the Bell violation $I^{\,\rm run} = \vec{\beta} \cdot \vec{P}^{\,\rm run}$ is given by $\sigma^2 = \vec{\beta}^{\top} ~ \Sigma ~ \vec{\beta}$: it depends on $\vec{\beta}_\mathrm{SI}$, which can be optimized. 

\subsection{Construction of the optimal variant}
\label{sec:conoptvar}

Finding the optimal variant of a Bell inequality boils down to the minimization of $\sigma$ over the variables $\vec{\beta}_\mathrm{SI}$. We prove in \ref{app:Derivation} that the variant of the inequality $\vec{\beta}^*$ whose value has minimal variance is given by:
\begin{equation}\label{eq:optimalvariant}
\vec{\beta}^* = \left (\bar{\Pi} -\Pi\left(\Pi \,\Sigma \,\Pi+\bar{\Pi}\right)^{-1}\Pi \, \Sigma \,\bar{\Pi} \ \right ) \vec{\beta},
\end{equation}
where $\Pi$ is the orthogonal projection over the signaling space $V_\mathrm{SI}$ (given in \ref{app:Derivation} and \tablename{~\ref{tab:basis}}) and $\bar{\Pi}=\id-\Pi$. If the matrix $\left(\Pi \,\Sigma \,\Pi+\bar{\Pi}\right)$ is not invertible, the inverse can be replaced by the Moore-Penrose pseudo-inverse.

Note that the implementation of Eq.~\ref{eq:optimalvariant} , which is the solution to a minimization problem of a quadratic function (standard problem in optimization), can be done easily and efficiently.

\subsection{Implementation on simulated loophole-free Bell experiments}

To illustrate the impact of our method, we study the robustness the CH, CHSH and Eberhard inequalities to the effects of finite statistics. Then, we derive the optimal Bell inequality. The original expressions have been shifted and rescaled so that the local bound is at $0$ and the maximal quantum violation at $2(\sqrt 2 -1)$. Specifically:\footnote{Our convention is $
\vec{\beta} = 
\left(
\begin{array}{rr|rr}
\beta_{000^{*}0^*} & \beta_{010^*0^*} & \beta_{000^*1^*} & \beta_{010^*1^*} \\
\beta_{100^*0^*} & \beta_{110^*0^*} & \beta_{100^*1^*} & \beta_{110^*1^*} \\ \hline
\beta_{001^*0^*} & \beta_{011^*0^*} & \beta_{001^*1^*} & \beta_{011^*1^*}  \\
\beta_{101^*0^*} & \beta_{111^*0^*} & \beta_{101^*1^*} & \beta_{111^*1^*}   \\
\end{array}
\right).
$}
\begin{itemize}
\item The CHSH inequality~(\cite{Hensen2015}) is shifted by $-\frac{1}{8}\vec{1}$. We consider $\vec{\beta}^{CHSH}=\vec{\beta}^{CHSH}_{\cite{Hensen2015}}-\frac{1}{2}\vec{1}$. 
$$
\vec{\beta}^{CHSH} = 
\left(
\begin{array}{rr|rr}
0.5 & -1.5 & 0.5 & -1.5 \\
-1.5 & 0.5 & -1.5 & 0.5 \\ \hline
0.5 & -1.5 & -1.5 & 0.5  \\
-1.5 & 0.5 & 0.5 & -1.5   \\
\end{array}
\right)
$$
\item The CH inequality~(\cite{Clauser1974}) has its marginal terms $p^\mathrm{A}_{a|x}$ and $p^\mathrm{B}_{b|y}$ interpreted as $p^\mathrm{A}_{a|x}= \sum_{b} p_{ab|x0}$ and  $p^\mathrm{B}_{b|y} = \sum_{b} p_{ab|0y}$, respectively. To obtain the same maximal violation, the inequality is scaled by: $\vec{\beta}^{CH}=4 \vec{\beta}^{CH}_{\cite{Clauser1974}}$.
$$
\vec{\beta}^{CH} = 
\left(
\begin{array}{rr|rr}
-4 & -4 & 4 & 0 \\
-4 & 0 & 0 & 0 \\ \hline
4 & 0 & -4 & 0  \\
0 & 0 & 0 & 0   \\
\end{array}
\right)
$$
\item The EH inequality~(\cite{Shalm2015}) corresponds to the CH inequality with marginal terms interpreted as $p^\mathrm{A}_{a|x}= \sum_{b} p_{ab|x1}$ and  $p^\mathrm{B}_{b|y} = \sum_{b} p_{ab|1y}$. It is scaled by  $\vec{\beta}^{EH}=4 \vec{\beta}^{EH}_{\cite{Shalm2015}}$. 
$$
\vec{\beta}^{EH} = 
\left(
\begin{array}{rr|rr}
0 & 0 & 0 & -4 \\
0 & 4 & 0 & 0 \\ \hline
0 & 0 & 0 & 0  \\
-4 & 0 & 0 & -4   \\
\end{array}
\right)
$$
\end{itemize}

To illustrate our method, we simulate two types of experiments. The first simulation is inspired by the experiment described in~\cite{Hensen2015}. We suppose that the source deterministically provides pairs of entangled photons. In this configuration, the vacuum component is negligible. For the second case, we assume that the source is based on the spontaneous parametric down conversion process, as employed in~\cite{Shalm2015}. In this last case, the vacuum component is predominant. As described in \ref{app_prob}, we use the features given in~\cite{Shalm2015,Hensen2015} to generate the following probability distributions: 
\begin{equation}
 \label{eq_prob_delft}
\scalebox{0.85}{$
\vec{P}^{(1)} = 
\left(
\begin{array}{cc|cc}
0.39 &  0.09 & 0.35 & 0.13 \\
 0.08 & 0.44 & 0.12 & 0.40 \\ \hline
0.39 & 0.09 & 0.10 & 0.38  \\
0.08 & 0.44 & 0.37 &  0.15  \\
\end{array}
\right)
%
\vec{P}^{(2)} = 
\left(
\begin{array}{cc|cc}
1-4.0\times 10^{-5} & 1.0\times 10^{-5} & 1-9.8\times 10^{-5} & 8.7\times 10^{-6} \\
 9.7\times 10^{-6} & 2.0\times 10^{-5} & 6.7\times 10^{-5} & 2.2\times 10^{-5} \\ \hline
1 - 9.8\times 10^{-5}& 6.8\times 10^{-5} & 1-1.8\times 10^{-4} & 9.0\times 10^{-5}  \\
8.3\times 10^{-6} & 2.2\times 10^{-5} & 8.9\times 10^{-5} & 4.7\times 10^{-7}   \\
\end{array}
\right)
$}
\end{equation}
To see the effect of the finite number of trials on the violation, we performed a Monte-Carlo type simulation. Given a number of trials, the algorithm returns a simulated number of detection events for the sixteen combinations of settings and outcomes.

\begin{figure}
\center
\includegraphics[width=0.47\columnwidth]{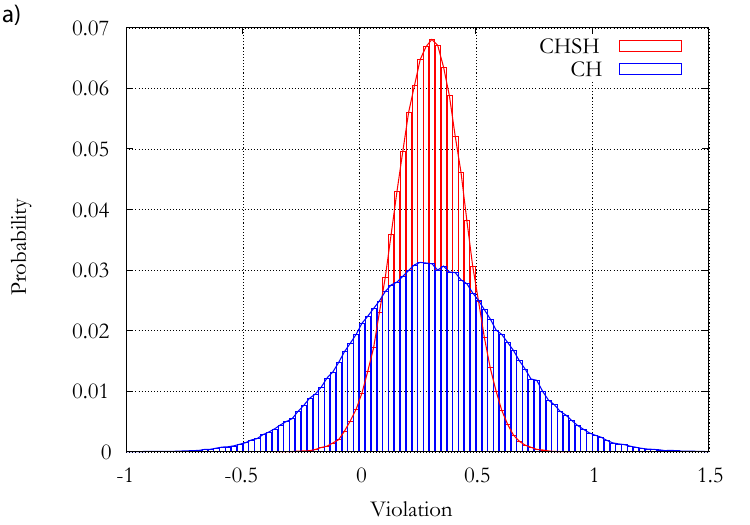}
\includegraphics[width=0.47\columnwidth]{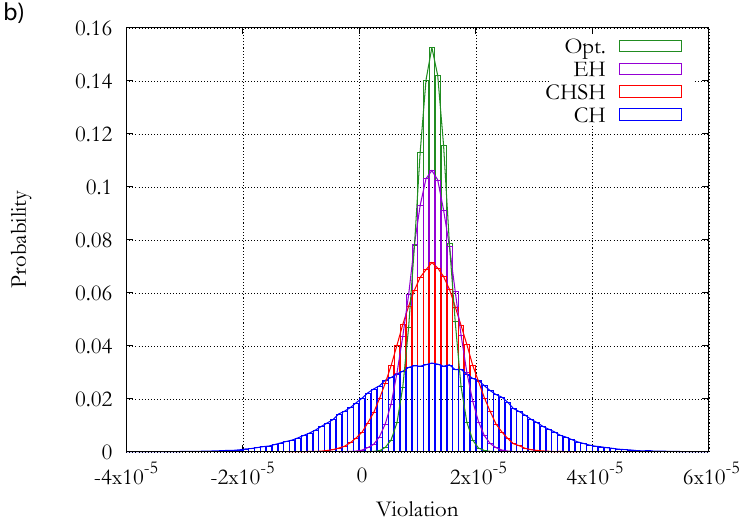}
\caption{a. Violation of Bell inequality for $\vec{P}^{(1)}$~(\cite{Hensen2015}). We observe same average violation $\left < I \right > = 0.302$ but different standard deviation $\sigma^{\,\rm CHSH} = 0.211$ and $\sigma^{\,\rm CH} = 0.464$. b. Violation of Bell inequality for $\vec{P}^{(2)}$~(\cite{Shalm2015}). We observe same average violation $\left < I \right > = 1.25\times10^{-5}$ but different standard deviation $\sigma^{\,\rm CHSH} = 5.65\times10^{-6}$, $\sigma^{\,\rm CH} = 1.20\times10^{-5}$, $\sigma^{\,\rm EH} =3.72\times10^{-6}$ and $\sigma^{\, \rm opt} = 2.60\times10^{-6}$. In both cases, we renormalized the inequalities such as the bound for the violation is $0$ (and not $2$ for CHSH) and such as the maximal quantum violation is the same.}
\label{fig:plotviolation}
\end{figure}

 The histogram distributions obtained for the two experiments after 200'000 simulated runs are given in \figurename{~\ref{fig:plotviolation}}. The \figurename{~\ref{fig:plotviolation}a} shows the violation distribution obtained when we perform 245 steps with the probability distribution $\vec{P}^{(1)}$ given in Eq.~\eqref{eq_prob_delft}. The CHSH and CH inequalities give the same amount of violation in average without the same standard deviation. As CHSH and CH only differs on their signaling component, this clearly show the influence of the signaling part on the statistical fluctuations around the mean violation. In this case, because the distribution obeys a particular symmetry (up to experimental imperfections), the inequality variant with minimal variance is CHSH (see~\ref{app:symarg}).

We also observe this effect with a non-maximally entangled photon pair source described by the probability distribution $\vec{P}^{(2)}$ given in Eq.~\eqref{eq_prob_delft}.  \figurename{~\ref{fig:plotviolation}b} shows the violation distributions obtained for $176\times 10^6$ steps. As before, we obtain for the three inequalities, CHSH, CH, and EH, the same amount of violation in average with different standard deviations.  Moreover, by using Eq.~\eqref{eq:optimalvariant}, we found an inequality that exploits the noise correlation between the probabilities to reduce the influence of the statistical error. This optimal inequality is: 
\begin{equation}
\vec{\beta}^\text{opt} = 
\left(
\begin{array}{rr|rr}
0 & -1.5 & 0 & -0.5 \\
-1.5 & 1 & -2.5 & 1 \\ \hline
0 & -2.5 & 0 & 0.5  \\
-0.5 & 1 & 0.5 & -3   \\
\end{array}
\right)
\end{equation}
As expected, we see in \figurename{~\ref{fig:plotviolation}b} that this inequality gives the same average violation with an improved standard deviation. This upgrades the violation, from $1.0\,\sigma$ with CH and $3.4\,\sigma$ with EH (considered in \cite{Shalm2015}) to $4.8\,\sigma$} with the optimal inequality.

\section{Physical properties of scenarios naturally emerge from the symmetry group}
\label{sec:decompo}

We saw in \sectionname{\ref{sec:labelingequivprinciple}} that the description $\vec{P}$ of a device-independent setup is made using an arbitrary labeling convention. Moreover, there exists a group $G$ of transformations, the relabelings, between those conventions and $V \ni \vec{P}$ is a representation of that group. In the following, we decompose it into irreducible components and identify each irreducible component with a corresponding physical property. By construction, those properties are independent from the labeling convention in use (for an introduction to representation theory, see~\ref{app:rep}).

\begin{theorem-non}
\label{theorem}
In the scenario $(2,2,2)$, the representation $V$ of $G$ is reducible, with orthogonal irreducible nonequivalent components:
\begin{equation}
V = V_\mathrm{NO_1} \oplus V_\mathrm{NO_2} \oplus V_\mathrm{NO_3} \oplus V_\mathrm{marg} \oplus V_\mathrm{corr} \oplus V_\mathrm{SI} \eegale{\text{coarse}}{\text{graining}} V_\mathrm{NO} + V_\mathrm{NS} + V_\mathrm{SI},
\end{equation}
where the subspaces have their bases and associated physical meanings given in \tablename~\ref{tab:components}. This decomposition is unique as the finest decomposition that preserves the group structure. Any vector $\vec{P}$ can be decomposed as:
\begin{equation}
\vec{P} = \vec{P}_{\mathrm{NO_1}} + \vec{P}_{\mathrm{NO_2}} + \vec{P}_{\mathrm{NO_3}} + \vec{P}_{\mathrm{marg}} + \vec{P}_{\mathrm{corr}} + \vec{P}_{\mathrm{SI}} \eegale{\text{coarse}}{\text{graining}} \vec{P}_{\mathrm{NO}} + \vec{P}_{\mathrm{NS}} + \vec{P}_{\mathrm{SI}},
\end{equation}
such that for any $g \in G$, we have $\vec{P}^g=\vec{P}_{\mathrm{NO_1}}^g + \vec{P}_{\mathrm{NO_2}}^g + \vec{P}_{\mathrm{NO_3}}^g + \vec{P}_{\mathrm{marg}}^g + \vec{P}_{\mathrm{corr}}^g + \vec{P}_{\mathrm{SI}}^g$ with $\vec{P}_{\mathrm{LBL}}, \vec{P}_{\mathrm{LBL}}^g \in V_{\mathrm{LBL}}$ for all subspace labels $\mathrm{LBL}$. In particular, "$\vec{P}_{\mathrm{LBL}}\neq 0$" is independent from the labeling convention. This can be seen as a confirmation of our intuition: "the setup contains the physical property ${\mathrm{LBL}}$" is an assertion which is independent from the labeling convention.\\ 
The same decomposition applies to the coefficient vector $\vec{\beta}$ of a Bell inequality.
\end{theorem-non}

\begin{proof}
The irreducible decomposition of $V$ is constructed step by step in \ref{app:step}. 
\end{proof}
\begin{table}
\centering
\begin{tabular}{l l l l l}
  \hline
 Vector subspace & Basis & Physical meaning & $\mathrm{dim}$  \\
  \hline
  \hline
$V_{\mathrm{NO}} \equiv V_{\mathrm{NO_1}} \oplus V_{\mathrm{NO_2}} \oplus V_{\mathrm{NO_3}}$  &&  Normalization & $4$ \\  
\hspace{20pt}$V_{\mathrm{NO_1}}$& $\vec{Q}_{++++}$ &Normalization& $1$  \\ 
\hspace{20pt}$V_{\mathrm{NO_2}}$& $\vec{Q}_{+++-}, \vec{Q}_{++-+}$ &Normalization& $2$  \\ 
\vspace{5pt}
\hspace{17pt}$V_{\mathrm{NO_3}}$& $\vec{Q}_{++--}$ &Normalization& $1$  \\ 
$V_{\mathrm{NS}} \equiv V_{\mathrm{marg}} \oplus V_{\mathrm{corr}}$ &&  Nonsignaling  & $8$  \\ 
\hspace{20pt}$V_{\mathrm{marg}} (\equiv V_{\mathrm{marg}}^\mathrm{A} + V_{\mathrm{marg}}^\mathrm{B})$ & &  Marginals  & $4$   \\  
\hspace{40pt}$(V_{\mathrm{marg}}^\mathrm{A})$ & $\vec{Q}_{-+\pm+}$ &  Alice marginals  & $2$  \\  
\vspace{5pt}  
\hspace{37pt}$(V_{\mathrm{marg}}^\mathrm{B})$ & $\vec{Q}_{+-+\pm}$ &  Bob marginals  & $2$  \\  
\vspace{10pt}  
\hspace{17pt}$V_{\mathrm{corr}}$ & $\vec{Q}_{--\pm\pm}$ &  Correlations & $4$  \\
$V_{\mathrm{SI}} (\equiv V_{\mathrm{SI}}^{ \rightarrow \mathrm{B}} + V_{\mathrm{SI}}^{ \rightarrow \mathrm{A}})$ & &  Signaling & $4$  \\  
\hspace{20pt}$(V_{\mathrm{SI}}^{\rightarrow \mathrm{B}})$ & $\vec{Q}_{+--\pm}$ &  Signaling A $\rightarrow$ B  & $2$ \\  
\hspace{20pt}$(V_{\mathrm{SI}}^{\rightarrow \mathrm{A}})$ & $\vec{Q}_{-+\pm-}$ &  Signaling B $\rightarrow$ A  & $2$  \\  

  \hline
\end{tabular}
\caption{
\label{tab:components}
Irreducible components appearing in the probability distributions of bipartite Bell scenarios with their dimension and basis ($Q_{ijkl}(ab|xy) = i^a j^b k^x l^y$ for $i,j,k,l=\pm 1$, see \tablename~\ref{tab:basis}).}
\end{table}
We now interpret the elements that appear in the decomposition, identifying them to intuitive physical properties. We illustrate it with setups whose behavior $\vec{P}$ contain specific physical behavior (\tablename~\ref{tab:prob}):
\begin{itemize}

\item Normalization subspace $V_\mathrm{NO}$:

This subspace corresponds to the physical property ``$\vec{P}$ is properly normalized'', corresponding to the four normalization conditions~\eqref{eq:normalization}. For any normalized outcome distribution, $\vec{P}_{\mathrm{NO_1}}=\frac{1}{4}$ and $\vec{P}_{\mathrm{NO_2}}=\vec{P}_{\mathrm{NO_3}}=0$. 

In the setup UNI (\tablename{~\ref{tab:prob}}), Alice and Bob throw two unbiased coins for any setting. This results in the uniform distribution $\vec{P}_\mathrm{UNI} \in V_\mathrm{NO}$.

\item Marginals subspace $V_\mathrm{marg}$:

This subspace corresponds to the marginals of Alice and Bob respectively independent from Bob and Alice choice of setting. We further split it into $V_\mathrm{marg} \equiv V_\mathrm{marg}^\mathrm{A} + V_\mathrm{marg}^\mathrm{B}$ by discerning the roles of Alice and Bob, otherwise mixed by the group. 
The physical property ``Alice and Bob have uniform marginals'' ($p_{a|x} = p_{b|y} = 1/2$) is preserved under relabelings and corresponds to $\vec{P}_\mathrm{marg} = 0$. In the correlator notation~\cite{Brunner2014}, it corresponds to the $\big < A_x \big > = \frac{1}{2} \sum_{aby} (-1)^a p_{ab|xy}$ and $\big < B_y \big > = \frac{1}{2} \sum_{abx} (-1)^b p_{ab|xy}$ (providing that the distribution of measurement settings is itself unbiased for signaling setups).

In the setup BIAS (\tablename{~\ref{tab:prob}}), Alice throws a fair coin and Bob a $\frac{1}{4}/\frac{3}{4}$ biased coin  for any setting. We observe that $\vec{P}_\mathrm{BIAS} = \vec{P}_\mathrm{NO} + \vec{P}_\mathrm{marg}$.

\item Correlation subspace $V_\mathrm{corr}$:

This subspace corresponds to the correlation between Alice and Bob outcomes, i.e. whether $a = b$. Outcomes distributions $\vec{P} = \vec{P}_\mathrm{NO} + \vec{P}_\mathrm{corr}$ are symmetric under relabeling of all outcomes ($a \rightarrow 1 - a$ and $b \rightarrow 1 - b$), and are optimally tested by CHSH (see~\ref{app:symarg}). In the correlator notation, it corresponds to the elements $\big < A_x B_y \big > = \sum_{ab} (-1)^{ab} p_{ab|xy}$. 

In the setup BW (\tablename{~\ref{tab:prob}}), an object with color either black or white (with equal probability $\frac{1}{2}$) is cut it in two parts which are sent to Alice and Bob. The outcome measured by Alice and Bob attributes a number $a=0,1$ and $b=0,1$ to the colors. The resulting distribution, $\vec{P}_\mathrm{BW} = \vec{P}_\mathrm{NO} + \vec{P}_\mathrm{corr}$, contains correlations but is nonsignaling with uniform marginals. 

\item Signaling subspace $V_\mathrm{SI}$:

This subspace corresponds to the physical property ``the outcome distribution is signaling''. As for $V_\mathrm{marg}$, we discern the roles of Alice and Bob to split $V_\mathrm{SI} \equiv V_{\mathrm{SI}}^{\rightarrow \mathrm{B}} + V_{\mathrm{SI}}^{ \rightarrow \mathrm{A}}$, representing respectively the signaling from Alice to Bob (when $\vec{P}_{\mathrm{SI}}^{ \rightarrow \mathrm{B}} \ne 0$) and from Bob to Alice (when $\vec{P}_{\mathrm{SI}}^{ \rightarrow \mathrm{A}} \ne 0$). To be more explicit, $\vec{P}_{\mathrm{sign}}^{ \rightarrow \mathrm{A}}$ contains strictly the information about the dependence of Alice's outcome on Bob's choice of setting. 

In the setup SIG (\tablename{~\ref{tab:prob}}), Alice measurement setting is Bob outcome and Alice throws an unbiased coin. It contains only signaling from Alice to Bob. The resulting distribution is purely signaling: $\vec{P}_\mathrm{SIG} = \vec{P}_\mathrm{NO} + \vec{P}_\mathrm{SI}^{\rightarrow \mathrm{B}}$.
\end{itemize}
\tablename{~\ref{tab:prob}} also contains the setup OPTQ corresponding to the optimal quantum violation of CHSH and the setup PRBOX corresponding to the PR-box~\cite{Popescu1994}. Neither setup has marginal or signaling term. Their correlation terms $\vec{P}_\mathrm{corr}$ are collinear, with PRBOX more correlated than OPTQ. The last line NOISE gives the decomposition of the statistical fluctuations present when $\vec{P}^\mathrm{setup}$ is estimated using a finite number of samples. In general, there is noise in each of $V_\mathrm{marg}, V_\mathrm{corr}$ and $V_\mathrm{SI}$.
\begin{table}
\center
\begin{tabular}{|m{1.4cm}|m{3.8cm}|l|m{4.5cm}|}
  \hline
  Name & Nonzero components & $p_{ab|xy} $ & $\vec{P} $ \\
  \hline
  UNI & $V_\mathrm{NO}$ & $\frac{1}{4}$  &  $\frac{1}{4}\vec{Q}_{++++}$ \\
  \hline
  BIAS & $V_\mathrm{NO}$, $V_\mathrm{marg}^\mathrm{B}$ & $\frac{1}{2}.(\frac{1}{4}\delta_{b=0}+\frac{3}{4}\delta_{b=1})$ & $\frac{1}{4} \vec{Q}_{++++} - \frac{1}{8} \vec{Q}_{+-++}$  \\
  \hline
  SIG & $V_\mathrm{NO}$, $V_\mathrm{SI}^{\rightarrow \mathrm{B}}$ & $\frac{1}{2}\delta_{x=b}$  &  $\frac{1}{4} \vec{Q}_{++++} + \frac{1}{4} \vec{Q}_{+--+}$\\
  \hline
  BW & $V_\mathrm{NO}$, $V_\mathrm{corr}$ &  $\frac{1}{2} \delta_{a=b}$ &  $\frac{1}{4} \vec{Q}_{++++} + \frac{1}{4} \vec{Q}_{--++}$ \\
  \hline
  PRBOX & $V_\mathrm{NO}$, $V_\mathrm{corr}$  & $\frac{1}{2}\delta_{a+b=xy}$ &  $\frac{1}{4} \vec{Q}_{++++} + \frac{1}{8} \sum_{kl} \tau_{kl} \vec{Q}_{--kl}$ \\  
  \hline
  OPTQ & $V_\mathrm{NO}$, $V_\mathrm{corr}$ & $p_{ab|xy}^{\psi^-,\mathrm{th}}$ &  $\frac{1}{4} \vec{Q}_{++++} + \frac{1}{8 \sqrt{2}} \sum_{kl} \tau_{kl} \vec{Q}_{--kl}$ \\
  \hline
  NOISE & $V_\mathrm{marg}$, $V_\mathrm{corr}$, $V_\mathrm{SI}$ &  $\epsilon_{ab|xy}=p^\mathrm{exp}_{ab|xy} - p^\mathrm{setup}_{ab|xy}$  &  $\vec{\epsilon}_{\mathrm{marg}}$, $\vec{\epsilon}_{\mathrm{corr}}$, $\vec{\epsilon}_{\mathrm{SI}}$\\
    \hline
\end{tabular}
\caption{\label{tab:prob} List of setups used to identify the physical properties corresponding to irreducible representations, with NOISE corresponding to the statistical fluctuations studied in \sectionname{\ref{sec:optimal}}. For each, we list the nonzero components, their probability distribution $p_{ab|xy}$ and its decomposition on the subspaces (using $Q_{ijkl}(ab|xy) = i^a j^b k^x l^y$, see~\ref{app:step}). Here $\tau_{kl}=-1$ if $k=l=-1$ and $\tau_{kl}=+1$ otherwise. For the last line, one has $\vec{P}^{\mathrm{exp}}=\vec{P}^{\mathrm{setup}}+\vec{\epsilon}$.}
\end{table}

This theorem, whose generalization to all (n,m,k) scenarios will be given in a future note~\cite{tobepublished}, is a direct proof of the unicity of the decomposition made in~\cite{Rosset2014a}. In this work, Rosset \textit{et al.} defined a projection on the nonsignaling subspace that commutes with the relabelings. As shown here for the scenario (2,2,2), this projection is unique.

\section{Conclusion}\label{sec:conclusion}

The study of quantum nonlocality in the device-independent framework considers the labeling of parties, measurement settings and outcomes as pure convention. Transformations between labeling conventions are given by the relabeling group. We explored the structure of this relabeling group and of some of its representations. In particular, the decomposition of the coefficient space of a Bell inequality singles out the normalization and signaling subspaces. Arbitrarily many variants of a Bell inequality can be generated by changing the vector components corresponding to these subspaces; however, if the obtained inequalities are equivalent on nonsignaling distributions, their statistical properties differ in experimental use. We showed how to construct the inequality with minimal variance, which optimizes the number of standard deviations above the local bound in experimental works. We simulated the outcome distributions coming from two recent experiments~\cite{Hensen2015,Shalm2015} and discussed the optimality of variants such that CHSH, CH or Eberhard inequalities compared to our method. 

Recently, several authors~\cite{Bierhorst2015,Elkouss2015} introduced formal statistical tests to reject the locality hypothesis with a certain $p$-value threshold. These tests are heavily dependent on the inequality variant considered. We leave the generalization of our method to these tests as an open question.

To our knowledge, this work represents the first application of the representation theory of finite groups to the relabeling group of Bell scenarios. Here, we described the decomposition of the representation space corresponding to $p_{ab...|xy...}$. We leave the study of the other irreducible representations to future work. If, as we conjecture, the relabeling group is a fundamental object in the study of device-independent protocols, its representations should play a major role in other protocols. 

\section*{Acknowledgements}\label{sec:acknowledgements}
We thank Jean-Daniel Bancal, Nicolas Brunner and Richard D. Gill for stimulating discussions. Also, we thank Tomer Barnea, Florian Frowis and  Yeong-Cherng Liang for comments on this manuscript. 
This work was supported by the National
Swiss Science Foundation (SNSF), the European Research Council (ERC
MEC) and the SNSF Early Postdoc Mobility grant P2GEP2\_162060.

\appendix

\section{Generation of the probability distributions }
\label{app_prob}

\subsection{Deterministic  entangled state sources}

The probability distribution for the deterministic entangled state source is calculated from the experimental parameters given in~\cite{Hensen2015}. In this experiment, the authors entangled two NV centers using the electronic spin as a two dimensional system with basis states denoted as $\ket \uparrow$ and $\ket \downarrow$. They aim to prepare the two spins in a maximally entangled state of the form $\ket {\psi^-} = \left[ \ket{\uparrow \downarrow} -\ket {\downarrow \uparrow} \right]/\sqrt 2$, but due to imperfections in the state preparation, they generate a partially entangled state $\rho$ described by the density matrix: 
\begin{equation}
\rho = \dfrac{1}{2}\begin{pmatrix}
\lambda & 0 & 0 & 0 \\
0 & 1-\lambda & V & 0 \\
0 & V & 1-\lambda & 0 \\
0 & 0 & 0 & \lambda
\end{pmatrix},
\end{equation}
with $\lambda = 0.022 $ and $V = 0.873$.

The measurement apparatuses perform an imperfect projection along the Z axis. To take into account the measurement imperfections the projectors are defined by: 
\begin{equation}
\Pi^i_{0} =  \eta_+^i\Pi_{+} +  (1-\eta_-^i)\Pi_{-}, \qquad \Pi^i_{1} =  (1-\eta_+^i)\Pi_{+} +  (\eta_-^i)\Pi_{-},
\end{equation}
where $\Pi_{\pm} = (\id\pm \sigma_z)/2$. The readout fidelities are equal to $\eta_+^A  = 0.954$, $\eta_-^A  = 0.994$, $\eta_+^B = 0.939$, and $\eta_-^B  = 0.998$ for the measurement apparatuses of Alice and Bob. The outcome probabilities are then 
$p_{ab|xy} = {\rm Tr}\left[ (\Pi^A_a \otimes \Pi^B_b) R(\theta^A_x,\theta^B_y) \rho  R(\theta^A_x,\theta^B_y)^\dagger  \right]$
where $R(\theta^A_x,\theta^B_y)$ is the rotation operator with respect to Z applied on the two spins.  The authors found that the maximal violation of the CHSH inequality is obtained for the angles $\theta^A_0 = 0$, $\theta^A_1 = \pi/2$, $\theta^B_0 = -3\pi/4-\epsilon$, and $\theta^B_1 = 3\pi/4 +\epsilon$, with $\epsilon = 0.026 \pi$. With this parameter we obtained the probability distribution given in Eq.~\eqref{eq_prob_delft}. As expected, this probability distribution gives the same violation ($2.30$) of the CHSH-Bell inequality than the one they mention (corresponding to the shifted value $0.30 = 2.30 - 2$ given in \figurename{~\ref{fig:plotviolation}a}).

\subsection{Nondeterministic entangled state sources}

The nondeterministic entangled state source was inspired by~\cite{Shalm2015}, where two spontaneous parametric down conversion processes are used to generate polarization entangled photon pairs  of the form $\ket{\psi} =  0.961 \ket{HH} + 0.276\ket{VV}$.
Here, the main difference with the previous approach is that the source is based on a spontaneous processes. 
This means that most of the time no photons are generated and rarely one or several pairs are. 
The nonlinear photon-pair generation nonlinear process with a large number of modes is described by the operator: 
\begin{equation}
C_i(\mu_i) = e^{-\frac{\mu_i}{2}} \sum_{n=0}^4 \frac{\mu_i^{n/2}}{n!^{3/2}} (a^\dagger_i b^\dagger_i )^n,
\end{equation}
where $\mu_i$ is the mean number of photon pairs in the spatial modes $a$ and $b$: the former is given to Alice, the latter to Bob. 
To properly define the density matrix measured by the authors, we need to take into account all the photon losses, which can be applied to mode $a$ via the loss operator: 
$L_a(\eta) = \sum_{n=0}^4 (1-\eta)^{n/2} \eta^{a^\dagger a/2} \frac{a^n}{\sqrt{n!}}$, where $\eta$ is the transmission probability. 
The complete source, i.e.~with the two nonlinear processes and all the losses over all the modes, is depicted by the density matrix of the form $\rho = S\coh {0}{0} S^\dagger $ with $S(\mu_H, \mu_V, \eta_A, \eta_B) = L_{b_H}(\eta_B)L_{b_V}(\eta_B)L_{a_V}(\eta_A)L_{a_H}(\eta_A)C_H(\mu_H)C_V(\mu_V)$. 
To simulate the source proposed in~\cite{Shalm2015}, we used the parameters $\eta_A = 0.747$, $\eta_B = 0.756$, $\mu_V = \frac{\mu}{1+r^2}$, and $\mu_H = \frac{r^2\mu}{1+r^2}$ with $r=0.288$ and $\mu = 4\times 10^{-4}$. 

For the detection, only one non-photon-number-resolving detector is used on each side to measure the photons in the $H$ mode. 
The detection events correspond to outcome 1 and the non-detection to an outcome 0. The measurement apparatus (written here for Alice) is described by the operators: 
\begin{eqnarray}
\Pi^A_{1} = (\id - \coh {0}{0} )_{a_H} \otimes \id_{a_V, b_H, b_V} \text{, and }
\Pi^A_{0} = \coh {0}{0}_{a_H} \otimes \id_{a_V, b_H, b_V}.
\end{eqnarray}
The probabilities given in Eq.~\eqref{eq_prob_delft} are obtained by $ p_{ab|xy} = {\rm Tr}\left[ (\Pi^A_a  \Pi^B_b) R(\theta^A_x,\theta^B_y) \rho  R(\theta^A_x,\theta^B_y)^\dagger  \right]$  for the polarization measurement angles $\theta^A_1= -4.2^\circ$, $\theta^A_1= 25.9^\circ$, $\theta^B_0= -4.2^\circ$, and $\theta^B_1= 25.9^\circ$.

\section{Introduction to representation theory}
\label{app:rep}

Here, we give a brief introduction to the theory of finite group representations and its main results we use. For a more complete introduction, we adress the reader to~\cite{Serre}.

A group of symmetries $G$ corresponds to a finite set of transformations $g$ with an algebraic group structure: identity, inverse, composition in $G$.\\
\textit{In our case, these symmetries will be build from permutations of the settings, outputs and parties}.\\
While group $G$ contains the abstract structure of the transformations, its concrete effects are seen on a  representation of $G$ on a vector space $V$, for which any $g\in G$ is associated with a linear invertible transformation of V ($v\in V\mapsto v^g\in V$) where this association preserves the structure of $G$.\vspace{0.1cm}\\
\textit{For example, imagine we are interested in the Bell scenario "Alice tosses a coin", i.e. (1,1,2). $G$ is the group of transformations from one arbitrary labeling (e.g. Heads=0, Tails=1) of the result to any possible labeling ($G=\mathfrak{S}_2=\{Id,\rho\}$, where $Id$ is the identity map and $\rho$ exchanges 0 and 1). 
Then $G$ acts on $V=\mathbb{R}^2$ by: $Id$ does nothing and $\rho$ exchanges the coordinates (this is called the natural representation).}\vspace{0.1cm}\\
$W\subset V$ is said to be invariant by $G$ if for any $w\in W, g\in G$, we always have $w^g\in W$.
Then, if the only subspaces $W\subset V$ invariant by $G$ are $\{0\}$ and $V$, we say that this representation is irreducible and write: "$V$ is an irreducible representation of $G$".\vspace{0.1cm}\\
\textit{Back to Alice and her coin, it is easy to see that the representation is not irreducible: $v=\binom{1}{1}$ is a counter example as $v^{Id}=v^\rho=v$.}\vspace{0.1cm}\\
We then have the following result, at the foundation of representation theory:

\begin{theorem-non}
If $G$ is finite, there is a finite number of irreducible representations ($W_1, W_2, ...$). Any representation $V$ of $G$ can be uniquely decomposed as an orthogonal sum of irreducible representations $V\egale{G}W_{i_1}\oplus...\oplus W_{i_p}$.\\
At the level of the vectors, any $v\in V$ can be uniquely decomposed as $v=w_1+...+w_p$ with $w_k\in W_{i_k}$. Then $\forall g\in G, v^g=w_1^g+...+w_p^g$ with $w_k^g\in W_{i_k}$.
\end{theorem-non}

This result, which looks quite technical, has an important physical meaning. If our formalism describes the physical reality in a vector space $V$, if this description is redundant and if $G$ contains the transformations between the different descriptions of the same physical reality, then our formalism contains some degeneracy (being the choice of some specific $g\in G$). Decomposing $V$ into irreducible representations is a way to remove this degeneracy. Most properties of $v\in V$ are not universal. For example, "first coordinate of $v$ is 0" is often an arbitrary property, as it usually depends on the choice of $g\in G$ (in our case, the choice of labeling). However, some properties of the $w_k$ components are universal: "$w_1=0$" is independent of the choice of $g\in G$.\\
We adopt this notation: 
If $V$ decomposes into $W_1$ and $W_2$ under $G$, we write $V\egale{G}W_1\oplus W_2$.
\vspace{0.1cm}\\
\textit{With a (1,1,2) Bell scenario, there are two irreducible representations of dimension one, $W_1=t$ (trivial irreducible representation) and $W_2=s$ (sign irreducible representation), with $\rho$ acting as $Id$ over $t$ and as $-Id$ over $s$. Here $V=t\oplus s$, with $t$ generated by $\epsilon_+=\binom{1}{1}$ and $s$ generated by $\epsilon_-=\binom{1}{-1}$. The two physical properties are the norm ($t$) and the bias of the coin ($s$).}\vspace{0.1cm}\\ 
In the following, we apply these results to our specific case. The symmetry group is $G$, the group of relabelings. We are then interested in the "relabeling-independent" physical properties of a Bell setup. As we will see, there are only four such properties: the norm, the marginals, the correlations and the signaling contribution. 

\section{Step by step construction of the $(2,2,2)$ scenario along with its symmetries}
\label{app:step}

\begin{figure}
\center
\includegraphics[width=0.8 \columnwidth]{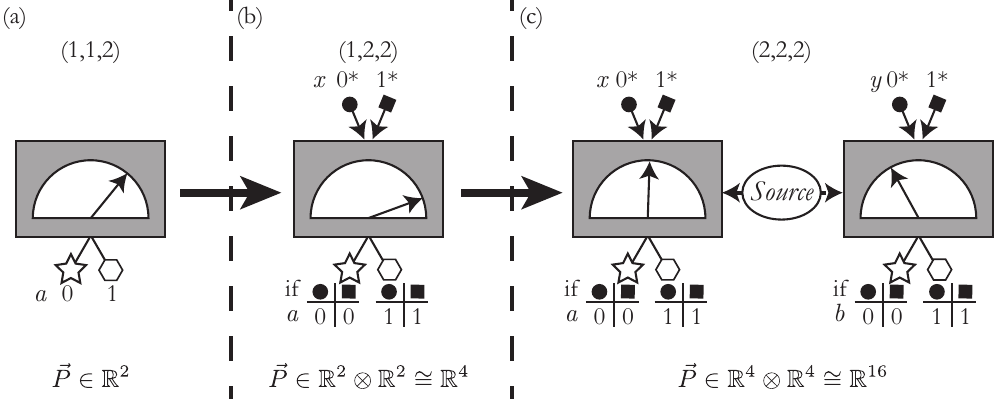}  
\caption{
Steps of the construction of the scenario $(2,2,2)$ made in~\ref{app:step}. We naturally go from scenario $(1,1,2)$ to scenario $(1,2,2)$ by adding two measurement settings, and to $(2,2,2)$ by adding two parties.
}
\label{fig_step}
\end{figure}

We prove here the \theoremname{} of \sectionname{\ref{sec:decompo}} by constructing the scenario, the symmetry group and its representation. As presented in~\figurename{~\ref{fig_step}}, the construction is done piece by piece, adding successively outcomes, settings and parties, while keeping track of the decomposition of the representation into irreducibles. We start with the scenario $(1,1,2)$ with two outcomes, then add the settings to obtain the scenario $(1,2,2)$ and lastly add parties, obtaining the full $(2,2,2)$ scenario. We give here the main ingredients of the proof, insisting on its physical meaning. For more technical details and a complete inventory of the irreducible representations of wreath products groups, see~\cite{CST}.
\\
Let us first recall our theorem:

\begin{theorem-non}
In the scenario $(2,2,2)$, the representation $V$ of relabeling group $G$ given by its action on $\vec{P}$ is reducible, with orthogonal irreducible nonequivalent components:
\begin{equation}
V = V_\mathrm{NO_1} \oplus V_\mathrm{NO_2} \oplus V_\mathrm{NO_3} \oplus V_\mathrm{marg} \oplus V_\mathrm{corr} \oplus V_\mathrm{SI},
\end{equation}
where the subspaces have their bases given in \tablename~\ref{tab:components}. 
\end{theorem-non}

As proved in the following, a basis for the above subspaces is given by the vectors $\vec{Q}_{ijkl}$, defined as
\begin{equation}
\label{eq:app:Q}
Q_{ijkl}(ab|xy) = i^a j^b k^x l^y, \qquad i,j,k,l= \pm 1.
\end{equation}
Any distribution $\vec{P}$ writes:
\begin{equation}
\label{eq:app:P}
\vec{P} = \sum_{ijkl=\pm 1}  \alpha_{ijkl} \vec{Q}_{ijkl}, \qquad \alpha_{ijkl} = \frac{1}{16} \sum_{abxy} i^a j^b k^x l^y p_{ab|xy}.
\end{equation}
While the individual basis vectors $\vec{Q}_{ijkl}$ are arbitrary, the subspaces themselves are uniquely specified by the action of the relabeling group.

\subsection{First step: scenario (1,1,2)}
\begin{table}
\center
\begin{tabular}{ c | c | c }
Elements of $G_{112}$  & $Id$:\begin{tabular}{c} $0 \rightarrow 0$ \\ $1 \rightarrow 1$ \end{tabular} &
$\rho$:\begin{tabular}{c} $0 \rightarrow 1$ \\ $1 \rightarrow 0$ \end{tabular} \\
\hline 
Transformation of \\
\includegraphics{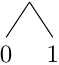}
& 
\includegraphics{ForestB11}
& 
\includegraphics{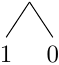}
\end{tabular}
\caption{Tree representations of the action of elements of $G_{\mathrm{112}}$.}
\label{Tree1}
\end{table}
We first start with a (1,1,2) Bell scenario, which corresponds to a physical device that outputs $\square$ or $\medcirc$, as represented in \figurename{~\ref{fig_step}~(a)}. We choose an arbitrary labeling, e.g. $\medcirc \equiv0$ and $\square \equiv 1$, to define the probability $p_{a}$ of obtaining the result 0 or 1. Then the probability vector is given by:
\begin{equation}
\vec{P}=\sum_{a}p_{a}\bm{e}_a=\left( \begin{array}{c}p_{0}\\ p_{1}\end{array} \right)\in V_{112}\equiv\mathbb{R}^2,
\end{equation} 
with $\bm{e}_0=\binom{1}{0},\bm{e}_1=\binom{0}{1}$ the canonical basis of $\mathbb{R}^2$. As represented in \tablename{~\ref{Tree1}}, the symmetry group is $G_{\mathrm{112}}\cong \mathfrak{S}_2=\{\Id, \rho\}$, where $\Id$ is the identity and $\rho$ exchanges $0$ and $1$, i.e. exchanges $e_0$ and $e_1$.

We now show how to express $\vec{P}$ in a more convenient way by giving the decomposition of $V_{112}$ into irreducible representations under $G_{112}$.
Defining $\bm{\epsilon}_\pm=\bm{e}_0\pm \bm{e}_1$, one can write $\vec{P}$ as $\vec{P}=\vec{P}_{\mathrm{NO}}+\vec{P}_{\mathrm{phy}}$ where $\vec{P}_{\mathrm{NO}}=\alpha_+ \bm{\epsilon}_+$  is the normalization part of $\vec{P}$ (with $\alpha_+=\frac{1}{2}(p_0+p_1)=\frac{1}{2}$) and $\vec{P}_{\mathrm{phy}}=\alpha_- \bm{\epsilon}_-$ is the physical part of $\vec{P}$ (with $\alpha_-=\frac{1}{2}(p_0-p_1)$). $\vec{P}_{\mathrm{phy}}$ is just the bias between the two outcomes.
The action of the elements of $G_{112}$ over $\vec{P}$ is then very easy to compute: $\Id$ does nothing, and $\rho$ changes the sign of $\vec{P}_{\mathrm{phy}}$. Then, we write $V_{\mathrm{112}}\egale{G_{112}}t\oplus s$, where $t$ (respectively $s$) is the vector space generated by $\{\epsilon_+\}$ (respectively $\{\epsilon_-\}$). The important properties of $t$ and $s$ are that they are stable under the action of elements of $G_{112}$. More precisely, $t$ (the trivial irreducible representation of $G_{112}$) is left invariant by any element of $G_{112}$ and $s$ (the sign irreducible representation of $G_{112}$) is such that $\rho$ acts as $-\Id$.

\subsection{Second step: scenario (1,2,2)}

As represented in \figurename{~\ref{fig_step}~(b)}, we extend our scenario to a (1,2,2) Bell scenario by adding two additional settings (e.g. with the arbitrary labeling $\medbullet\equiv0$ and $\blacksquare\equiv1$). When we label the experiment, we may decide on the labeling of the outcomes depending on the choice of measurement settings, which leads to eight possible labelings. After the labeling is chosen, we can define the probability $p_{a|x}$ to obtain an outcome $a$ given the setting $x$:
\begin{equation}\label{equation122}
\vec{P}=\sum_{ax}p_{a|x} ~ \bm{e}_a\otimes \bm{e}_x^* =\left( \begin{array}{c}p_{0|0^*}\\ p_{1|0^*}\\ p_{0|1^*} \\ p_{1|1^*}\end{array} \right)\in V_{\mathrm{122}}\equiv\mathbb{R}^2\otimes {\mathbb{R}^2}^*\cong\mathbb{R}^4,
\end{equation} 
where $\{ \bm{e}_a\}_{a=0,1}$ (respectively $\{ \bm{e}_x^*\}_{x=0,1}$) is the canonical basis of $\mathbb{R}^2$ (respectively ${\mathbb{R}^2}^*$). Note that we changed the usual convention for the tensor product enumeration, changing the value of subscript $a$ first. Again, we use an asterisk for the inputs, to distinguish them from the outputs.  An element of $G_{\mathrm{122}}$ is defined by one permutation of the setting $x$, and two permutations of the outcomes (one for each setting): two elements can be seen in \tablename{~\ref{Tree2}}. More formally, $G_{\mathrm{122}}=\mathfrak{S}_2 \wr \mathfrak{S}_2^*$ is the wreath product (written $\wr$) of $G_{112}\cong \mathfrak{S}_2$ (permutation of the settings) with $\mathfrak{S}_2^*$ (permutation of the outcomes) and $V_{\mathrm{122}}$ is the imprimitive representation constructed from $V_{\mathrm{112}}$ (see \ref{app:wreath}).

\begin{table}
\center
\begin{tabular}{c|cc|cc}
Some elements of $G_{122}$  & $\rho$:\begin{tabular}{c}$0^* \rightarrow 0^*$ 
 \\ $1^* \rightarrow 1^*$
 \end{tabular} &\begin{tabular}{l}$0\rightarrow 1$ 
 \\ $1\rightarrow 0$\end{tabular} &
$\rho$:\begin{tabular}{c} $0^* \rightarrow 1^*$ 
 \\ $1^* \rightarrow 0^*$\end{tabular}&
\begin{tabular}{l}
if $0^*$: \begin{tabular}{c}  \footnotesize $0 \rightarrow 1$ \\ \footnotesize $1 \rightarrow 0$ \end{tabular}\\ \hline
if $1^*$: \begin{tabular}{c} \footnotesize $0 \rightarrow 0$ \\  \footnotesize $1 \rightarrow 1$ \end{tabular} 
  \end{tabular} \\
\hline 
Transformation of \\
\includegraphics{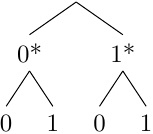}
&
\multicolumn{2}{c}{
\includegraphics{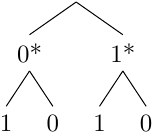}
}
&\multicolumn{2}{c}{
\includegraphics{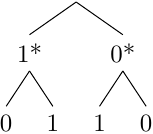}
}
\end{tabular}
\caption{Tree representation of the action of two elements of $G_{\mathrm{122}}$.}
\label{Tree2}
\end{table}

In order to decompose $V_{\mathrm{122}}$ into irreducible representations, we can first decompose it under the action of $G_{112}\subset G_{\mathrm{122}}$ only:
\begin{equation}
\mathbb{R}^2\otimes {\mathbb{R}^2}^* \egale{G_{112}}(t\oplus s)\otimes {\mathbb{R}^2}^*\egale{G_{112}}(t\otimes 0^*)\oplus(t\otimes 1^*)\oplus(s\otimes 0^*)\oplus(s\otimes 1^*),
\end{equation}
where $t\otimes x^*$ and $s\otimes x^*$ are generated by $\bm{\epsilon}_+\otimes \bm{e}^*_x$ and  $\bm{\epsilon}_-\otimes \bm{e}_x^*$, respectively. Then, to decompose under the full group $G_{\mathrm{122}}$, one has to take into account all possible permutations of the setting. Finally, we find: 
\begin{equation}
\mathbb{R}^2\otimes {\mathbb{R}^2}^* \egale{G_{122}} T \oplus S^* \oplus \phi,
\end{equation}
where $T\equiv t\otimes t^*$, $S^*\equiv t\otimes s^*$ and $\phi\equiv s\otimes 0^*+s\otimes 1^*$ are generated by $\{\bm{\epsilon}_+\otimes \bm{\epsilon}_+^*\}$, $\{\bm{\epsilon}_+\otimes \bm{\epsilon}_-^*\}$ and $\{\bm{\epsilon}_-\otimes \bm{\epsilon}_+^*$, $\bm{\epsilon}_-\otimes \bm{\epsilon}_-^*\}$, respectively. As expected, those three spaces are globally unchanged by any element $\rho\in G_{122}$.

In this decomposition, $\vec{P}$ can be written as $\vec{P}=\vec{P}_\mathrm{NO_1}+\vec{P}_\mathrm{NO_2}+\vec{P}_\mathrm{phy}$ where $\vec{P}_\mathrm{NO_1}= \alpha_{++} ~ \bm{\epsilon}_+\otimes\bm{\epsilon}_+^*\in T$ and $\vec{P}_\mathrm{NO_2}=\alpha_{+-} ~ \bm{\epsilon}_+\otimes \bm{\epsilon}_-^*\in S^*$ are the normalization parts and $\vec{P}_\mathrm{phy}= \alpha_{-+} ~ \bm{\epsilon}_-\otimes \bm{\epsilon}_+^*+\alpha_{--} ~ \bm{\epsilon}_-\otimes \bm{\epsilon}_-^*\in\phi$ is the physical part (with $\alpha_{ik}=\frac{1}{4}\sum_{ax} i^a k^x p_{a|x}$). Here $\alpha_{+-}$ represents the bias in the settings which is 0 for normalized probability vectors. In the physical part, $\alpha_{-+}$ is the sum of the outcome biases for the two possible settings while $\alpha_{--}$ is the difference of the outcome biases for the two possible settings.

\subsection{Last step: scenario (2,2,2)}
\begin{table}
\center
\begin{tabular}{c | c }
Element of $G_{222}$  &
$\rho$: 
\begin{tabular}{c}
	A $\rightarrow$ B 
	\begin{tabular}{c c } 
		\small 
		if A
		\begin{tabular}{c} 
			\footnotesize $0^* \rightarrow 0^*$ 
			 \\  \footnotesize$1^* \rightarrow 1^*$
		\end{tabular}&
			\begin{tabular}{l}
				\footnotesize if $0^*$: 
				\begin{tabular}{c}  
					\scriptsize $0 \rightarrow 1$ \\ \scriptsize $1 \rightarrow 0$ 
				\end{tabular}\\ \hline
				\footnotesize if $1^*$: 
				\begin{tabular}{c}
					\scriptsize $0 \rightarrow 1$ \\  \scriptsize $1 \rightarrow 0$ 
				\end{tabular} 
			\end{tabular} \\ \hline
		\end{tabular}\\ 
		\small B $\rightarrow$ A
		\begin{tabular}{c c }
		if B
		\begin{tabular}{c} 
			\footnotesize $0^* \rightarrow 1^*$ 
			 \\ \footnotesize $1^* \rightarrow 0^*$
		\end{tabular}&
			\begin{tabular}{l}
				\footnotesize if $0^*$: 
				\begin{tabular}{c}  
					\scriptsize $0 \rightarrow 1$ \\ \scriptsize $1 \rightarrow 0$ 
				\end{tabular}\\ \hline
				\footnotesize if $1^*$: 
				\begin{tabular}{c}
					\scriptsize $0 \rightarrow 0$ \\  \scriptsize $1 \rightarrow 1$ 
				\end{tabular} 
			\end{tabular}
		\end{tabular}
	\end{tabular}
\\
\hline 
Transformation of \\
\includegraphics{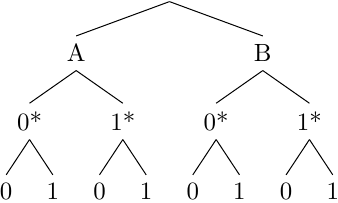}
& 
\includegraphics{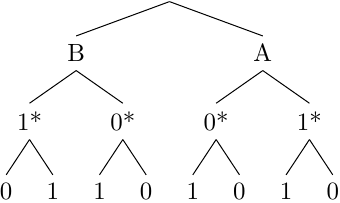}
\end{tabular}
\caption{\label{Tree3}Tree representation of the action of one elements of $G_{\mathrm{222}}$.}
\end{table}
To extend to a (2,2,2) Bell scenario, we consider two parties arbitrarily labeled as Alice and Bob, as represented in \figurename{~\ref{fig_step}~(c)}. The symmetry group $G_{\mathrm{222}}$ of this scenario is composed of 128 ways to relabel the setup: a specific element is represented in \tablename{~\ref{Tree3}}. For a specific setup, we can define the probability $p_{ab|xy}$ of obtaining the outcomes $(a,b)$ given that the settings are $(x,y)$. The associated probability vector is given by :
\begin{equation}
\vec{P}=\sum_{axby}p_{ab|xy}\bm{e}_a^A\otimes \bm{e}_x^X\otimes \bm{e}_b^B\otimes \bm{e}_y^Y =\left( \begin{array}{c}p_{00|0^*0^*}\\ p_{10|0^*0^*}\\ \vdots \\ p_{01|1^*1^*} \\ p_{11|1^*1^*}\end{array} \right)\in V_{222}\equiv(\mathbb{R}^2\otimes {\mathbb{R}^2}^*)_A \otimes (\mathbb{R}^2\otimes {\mathbb{R}^2}^*)_B\cong\mathbb{R}^{16}.
\end{equation}
(Here again, similarly to Eq.~(\ref{equation122}), we chose a subscript enumeration different from the usual one).
More formally, $G_{\mathrm{222}}$ is a wreath product: $G_{\mathrm{222}}= G_{\mathrm{122}} \wr \mathfrak{S}_2^{\mathrm{AB}}= \mathfrak{S}_2\wr \mathfrak{S}_2^{*}\wr \mathfrak{S}_2^{\mathrm{AB}}$ and $V_{\mathrm{222}}$ is the primitive representation constructed from $V_{\mathrm{122}}$. In order to decompose $V_{\mathrm{222}}$ into irreducible representations, we can first decompose independently each $(\mathbb{R}^2\otimes \mathbb{R}^*{^2})$ under $G_{\mathrm{122}}$:
\begin{eqnarray}
\nonumber V_{\mathrm{222}}
 &~=&(\mathbb{R}^2\otimes {\mathbb{R}^2}^*)_A \otimes (\mathbb{R}^2\otimes {\mathbb{R}^2}^*)_B\\
\nonumber {} &\egale{G_{\mathrm{122}}}& (T \oplus S^* \oplus \phi)_A\otimes (T \oplus S^* \oplus \phi)_B   \\
\nonumber {} &\egale{G_{\mathrm{122}}}& (T_A\otimes T_B) \oplus (T_A\otimes S_B^*) \oplus (S_A^*\otimes T_B) \oplus (S_A^*\otimes S_B^*)    \\
\nonumber {} &{}&\oplus (T_A\otimes \phi_B) \oplus (\phi_A\otimes T_B) \oplus (\phi_A\otimes \phi_B)  \oplus (S_A^*\otimes \phi_B) \oplus (\phi_A\otimes S_B^*) 
\end{eqnarray}
$G_{\mathrm{222}}$ recombines those components by allowing the exchange between A and B. Thus, we obtain:
\begin{multline}
V_{\mathrm{222}}\egale{G_{222}}(\overbrace{T_A\otimes T_B}^{V_{\mathrm{NO_1}}}) \oplus (\overbrace{T_A\otimes S_B^* +S_A^*\otimes T_B}^{V_{\mathrm{NO_2}}}) \oplus (\overbrace{S_A^*\otimes S_B^*}^{V_{\mathrm{NO_3}}}) \\
\oplus (\underbrace{T_A\otimes \phi_B}_{V_{\mathrm{marg}}^{\mathrm{B}}} + \underbrace{\phi_A\otimes T_B}_{V_{\mathrm{marg}}^{\mathrm{A}}}) \oplus (\underbrace{\phi_A\otimes \phi_B}_{V_{\mathrm{corr}}}) \oplus   (\underbrace{S_A^*\otimes \phi_B}_{V_{\mathrm{SI}}^{\mathrm{\rightarrow B}}}+ \underbrace{\phi_A\otimes S_B^*}_{V_{\mathrm{SI}}^{\mathrm{\rightarrow A}}})
\end{multline}
\figurename{~\ref{fig_permutirreps}} illustrates all the possible permutations of the six irreducible representations. Those irreducible representation are clearly nonequivalent.
In this decomposition, $\vec{P}$ can be written as:
\begin{equation}
\label{eq:decompo1}
\vec{P}=\vec{P}_\mathrm{NO1}+\vec{P}_\mathrm{NO2}+\vec{P}_\mathrm{NO3}+\vec{P}_\mathrm{marg}+\vec{P}_\mathrm{corr}+\vec{P}_\mathrm{SI}.
\end{equation}
Defining $\vec{Q}_{ijkl}= \bm{\epsilon}_{i,A}\otimes\bm{\epsilon}_{j,B}^*\otimes\bm{\epsilon}_{k,A}\otimes\bm{\epsilon}_{l,B}^*$ and $\alpha_{ijkl}=\frac{1}{16}\sum_{abxy} i^a j^b k^x l^y p_{ab|xy}$, we have:
\begin{equation}
\label{eq:decompo2}
\vec{P}=\sum_{ijkl}\alpha_{ijkl}\vec{Q}_{ijkl},
\end{equation}
The correspondence between Eq.~\eqref{eq:decompo1} and Eq.~\eqref{eq:decompo2} is direct and can be read in \tablename{~\ref{tab:components}}.

\begin{figure}
\begin{center}
\includegraphics{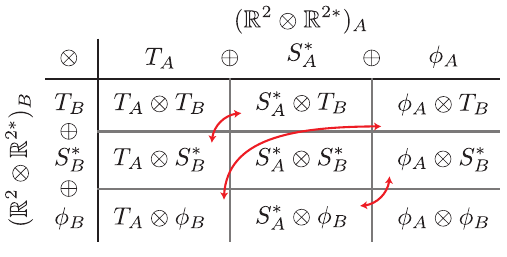} 
\caption{\label{fig_permutirreps}Decomposition into irreducible representations of the full scenario. Each cell represent an irreducible representation, under $G_{122}$. Red arrows indicates which irreducible representations are recombine under exchange of $A$ and $B$. Then, the full decomposition of the scenario under its relabeling group $G_{222}$ has six irreducible nonequivalent representations}
\end{center}
\end{figure}

\subsection{Physical interpretation of the irreducible representations} 
\begin{table}
\center
\begin{tabular}{r||c:cc:c|cc:cc|cccc|cc:cc}
 $\vec{Q} \hspace{0.4cm} i$   
       &+&+&+&+& +&+&-&-& -&-&-&-& +&+&-&-\\[-5pt]
 $j$   &+&+&+&-& -&-&+&+& -&-&-&-& -&-&+&+\\[-5pt]
 $k$   &+&+&-&+& +&+&+&-& +&-&+&-& -&-&+&-\\[-5pt]
 $l$   &+&-&+&+& +&-&+&+& +&+&-&-& +&-&-&-\\
$abxy$  &   &  &   & & & & & & & & &\\
  0000 &+&+&+& +& +&+&+&+&+&+&+&+&+&+&+&+\\
  1000 &+&+&+& +& +&+&-&-&-&-&-&-&+&+&-&-\\
  0100 &+&+&+& +& -&-&+&+&-&-&-&-&-&-&+&+\\
  1100 &+& +&+& +& -&-&-&-&+&+&+&+&-&-&-&-\\
  0010 &+&+&-&   -& +&+&+&-&+&-&+&-&-&-&+&-\\
  1010 &+&+&-&   -& +&+&-&+&-&+&-&+&-&-&-&+\\
  0110 &+& +&-&  -& -&-&+&-&-&+&-&+&+&+&+&-\\
  1110 &+& +&-&  -& -&-&-&+&+&-&+&-&+&+&-&+\\
  0001 &+&-&+&  -& +&-&+&+&+&+&-&-&+&-&-&-\\
  1001 &+&-&+&  -& +&-&-&-&-&-&+&+&+&-&+&+\\
  0101 &+&-&+&  -& -&+&+&+&-&-&+&+&-&+&-&-\\
  1101 &+& -&+& -& -&+&-&-&+&+&-&-&-&+&+&+\\
  0011 &+& -&-& +& +&-&+&-&+&-&-&+&-&+&-&+\\
  1011 &+& -&-& +& +&-&-&+&-&+&+&-&-&+&+&-\\
  0111 &+& -&-& +& -&+&+&-&-&+&+&-&+&-&-&+\\
  1111 &+& -&-& +& -&+&-&+&+&-&-&+&+&-&+&-\\

  &$\mathrm{NO}_1$&\multicolumn{2}{c:}{$\mathrm{NO}_2$}& $\mathrm{NO}_3$&   \multicolumn{2}{c:}{$\mathrm{marg^B}$}&\multicolumn{2}{c|}{$\mathrm{marg^A}$}& \multicolumn{4}{c|}{$\mathrm{corr}$}&\multicolumn{2}{c:}{$\mathrm{SI^{\rightarrow B}}$}&\multicolumn{2}{c}{$\mathrm{SI^{\rightarrow A}}$}
\end{tabular}

\caption{
\label{tab:basis} 
Coordinates of the orthogonal basis $\vec{Q}_{ijkl}$ and corresponding subspaces, as detailed in Eqs.~\eqref{eq:app:Q} and~\eqref{eq:app:P}. In this basis, $\vec{P}=\sum_{ijkl}\alpha_{ijkl}\vec{Q}_{ijkl}$ with $\alpha_{ijkl}=\frac{1}{16} \sum_{abxy} i^a j^b k^x l^y p_{ab|xy}$. Here "+" corresponds to "1" and "-" corresponds to "-1"}
\end{table}
We now interpret the subspaces of the \theoremname{} present in the decomposition~\eqref{eq:decompo1}.

\begin{itemize}

\item Normalization subspace $V_\mathrm{NO}$:

 Let $\vec{P}_{\mathrm{NO}}\in V_{\mathrm{NO}}$ be the component of $\vec{P}$ over the first three subspaces $V_{\mathrm{NO_1}}$, $V_{\mathrm{NO_2}}$ and $V_{\mathrm{NO_3}}$. As said in Section~\ref{sec:decompo}, it is easy to see that $\vec{P}_{\mathrm{NO}}$ is fixed by the four normalization conditions on $\vec{P}$, $\sum_{ab}p_{ab|xy}=1$, which impose $\alpha_{\mathrm{++++}}=\frac{1}{4}$ and $\alpha_{\mathrm{+++-}}=\alpha_{\mathrm{++-+}}=\alpha_{\mathrm{++--}}=0$. Then $\vec{P}_{\mathrm{NO}}$ characterizes the normalization of $\vec{P}_{\mathrm{NO}}$ only. 

This can also be deduced from the way elements of $G$ act over $V_\mathrm{NO}$. For $M= \mathrm{A}, \mathrm{B}$, $T_M=(t\otimes t^*)_M$ and $S^*_M=(t\otimes s^*)_M$ contain $t$, the trivial representation of the outcomes: the subspaces containing only $T_M$ and $S^*_M$ are not affected by any relabeling of the outcomes. Therefore they correspond to a property of $\vec{P}$ which can be known even by someone who has no access to the outcomes, so $\vec{P}_{\mathrm{NO}}$ can only be about the normalization of $\vec{P}$.\footnote{Those three components have an interpretation when this decomposition is directly applied to the statistics of the experiment $\vec{N}=\{N(abxy)\}$, before the normalization Eq.~\eqref{normalization}. In this case, they give information about the $N(xy)$. $T_A\otimes T_B$ is about the total number of steps $N=\sum_{x,y} N(xy)$, $S^*_A\otimes T_B$ the bias on Alice's setting $\sum_{x,y} (-1)^x N(xy)$, $T_A\otimes S^*_B$ the bias on Bob's setting $\sum_{x,y} (-1)^y N(xy)$ and $S^*_A\otimes S^*_B$ the correlations of the two biases $\sum_{x,y} (-1)^{x+y} N(xy)$.}

\item Marginal subspace $V_\mathrm{marg}$:

 $V_\mathrm{marg}^B = T_\mathrm{A}\otimes \phi_\mathrm{B}$ contains $T_\mathrm{A}=t_\mathrm{A}\otimes t_\mathrm{A}^*$. It is unchanged by any relabeling (setting or outcome) on Alice's side. Hence it corresponds to a property of $\vec{P}$ which can be known by Bob if he does not communicate with Alice and which does not depends on her setting. This can only be Bob's marginals. We call $\vec{P}_{\mathrm{marg}}^\mathrm{B}$ the projection of $\vec{P}$ over this subspace. Likewise, $V_\mathrm{marg}^A=\phi_\mathrm{A}\otimes T_\mathrm{B}$ corresponds to Alice's marginals. In terms of coefficients, we have $\vec{P}_{\mathrm{marg}}^\mathrm{A}=\alpha_{\mathrm{-+++}}\vec{Q}_{\mathrm{-+++}}+\alpha_{\mathrm{-+-+}}\vec{Q}_{\mathrm{-+-+}}$ and $\vec{P}_{\mathrm{marg}}^\mathrm{B}=\alpha_{\mathrm{+-++}}\vec{Q}_{\mathrm{+-++}}+\alpha_{\mathrm{+-+-}}\vec{Q}_{\mathrm{+-+-}}$.

\item Correlation subspace $V_\mathrm{corr}$:

 $\vec{P}_\mathrm{corr}$ contains the same information as the usual correlators. Indeed, the correlators $\big < A_x B_y \big >$ corresponds to the new orthogonal basis $\vec{E}_{xy}=\sum_{k,l}k^x l^y\vec{Q}_{1k1l}$. In this basis, we have:
\begin{equation}
\vec{P}_{\mathrm{corr}}=\sum_{xy} E_{xy} ~ \vec{E}_{xy},
\end{equation}
where $E_{xy}=\big < A_x B_y \big >$.

Looking at the way elements of $G$ act over $V_\mathrm{corr}$, we can show that it must correspond to correlations. Let us consider the relabeling of outputs $\rho\in G_{\mathrm{122}}$ of only one party (which exchange $0$ and $1$). This defines two relabelings of the $(2,2,2)$ scenario, $\rho_\mathrm{A}$ (respectively $\rho_\mathrm{B}$), which acts on Alice's (respectively Bob's) side only. Hence, $\phi_\mathrm{A}\otimes \phi_\mathrm{B}$ has the property of transforming the same way for $\rho_\mathrm{A}$ and $\rho_\mathrm{B}$: if $\vec{P}_{\mathrm{corr}}$ is the component of $\vec{P}$ over $V_{\mathrm{corr}}$, we have $\vec{P}_{\mathrm{corr}}^{\rho_\mathrm{A}}=\vec{P}_{\mathrm{corr}}^{\rho_\mathrm{B}}$. Now, as $\rho=\rho^{-1}$ (with two settings/outcomes), $(\vec{P}_{\mathrm{corr}}^{\rho_\mathrm{A}})^{\rho_\mathrm{B}}=\vec{P}_{\mathrm{corr}}$. Then $\vec{P}_{\mathrm{corr}}$ corresponds to a property of $\vec{P}$ where the relationship between Alice's and Bob's outcome is important, but the result of Alice by itself is not; this is the correlation part. 

\item Signaling subspace $V_\mathrm{SI}$:

$V_{\mathrm{SI}}^{ \rightarrow \mathrm{B}}=S^*_\mathrm{A}\otimes \phi_\mathrm{B}$ contains $S^*_\mathrm{A}=t\otimes s^*$, on which a permutation $\rho_\mathrm{A}^*$ of the setting on Alice's side acts as $-Id$. Therefore, if $\vec{P}_{SI}^\mathrm{B}$ is the component of $\vec{P}$ over this subspace, ${(\vec{P}_{SI}^\mathrm{B})}^{\rho_\mathrm{A}^*}=-\vec{P}_{SI}^\mathrm{B}$. Hence Alice's choice of settings has an influence on $\vec{P}_{SI}^\mathrm{B}$: as soon as $\vec{P}_{SI}^\mathrm{B}\neq 0$, Alice can signal to Bob. So this subspace must correspond to the signaling from Alice to Bob. Likewise, $V_{\mathrm{SI}}^{ \rightarrow \mathrm{A}}=\phi_\mathrm{A}\otimes S^*_B$ corresponds to the signaling from Bob to Alice. We have $\vec{P}_{\mathrm{SI}}^\mathrm{A}=\alpha_{\mathrm{-++-}}\vec{Q}_{\mathrm{-++-}}+\alpha_{\mathrm{-+--}}\vec{Q}_{\mathrm{-+--}}$ and $\vec{P}_{\mathrm{SI}}^\mathrm{B}=\alpha_{\mathrm{+--+}}\vec{Q}_{\mathrm{+--+}}+\alpha_{\mathrm{+---}}\vec{Q}_{\mathrm{+---}}$.

\end{itemize}

\section{Optimal form of a Bell inequality}

\subsection{Derivation}
\label{app:Derivation}

Here, we prove formula Eq.~\eqref{eq:optimalvariant}.
Remember that $\Pi$ is the projector over the signaling subspace, directly given by $\Pi=\frac{1}{16}\sum_{ijkl\in \mathrm{SI}} \vec{Q}_{ijkl}\vec{Q}_{ijkl}^\top$, where the $\vec{Q}_{ijkl}$ are given in Table~\ref{tab:basis}. We saw in Section~\ref{sec:spusignfinstat} that the mean of $I^{\,\rm run} = \vec{\beta} \cdot \vec{P}^{\,\rm run}$ is given by $\left < I^{\,\rm run} \right > = I^{\,\rm setup}$. We continue here with the computation of the variance $\sigma$:

\begin{equation}
\label{eq:derivvar}
\sigma^2 = \vec{\beta}^{\top} \, \Sigma \, \vec{\beta} =  \vec{\beta}_\mathrm{NOS}^{\top} \Sigma \, \vec{\beta}_\mathrm{NOS} + \vec{\beta}^{\top}_\mathrm{SI} \Sigma \, \vec{\beta}_\mathrm{SI} + 2\, \vec{\beta}^\top_\mathrm{NOS} \Sigma\, \vec{\beta}_\mathrm{SI},
\end{equation}
where $\vec{\beta}_\mathrm{NOS} = \vec{\beta}_\mathrm{NO} + \vec{\beta}_\mathrm{NS}$. With $\Pi$ the orthogonal projector over $V_\mathrm{SI}$ and $\bar{\Pi}=\id-\Pi$, we have $\vec{\beta}_\mathrm{SI}=\Pi \vec{\beta}$ and $\vec{\beta}_\mathrm{NOS}=\bar{\Pi} \vec{\beta}$ \\

Now, we search the Bell inequality variant with the minimal variance $\sigma^2$. This corresponds to finding the minimum $\vec{\beta}_\mathrm{SI}^*$ of the following function $f$ over the space $V_\mathrm{SI}$:
\begin{equation}
f(\vec{\beta}_\mathrm{SI}) = \sigma^2 = \vec{\beta}_\mathrm{NOS}^{\top} \Sigma \vec{\beta}_\mathrm{NOS} + \vec{\beta}^{\top}_\mathrm{SI} \Sigma \vec{\beta}_\mathrm{SI} + 2 \vec{\beta}^\top_\mathrm{NOS} \Sigma \vec{\beta}_\mathrm{SI}.
\end{equation}
As the covariance matrix $\Sigma$ is positive semidefinite by construction, the minimum is obtained by $\nabla_{\vec{\beta}^*_\mathrm{SI}} f$ null over $V_\mathrm{SI}$ i.e. $\Pi\nabla_{\vec{\beta}^*_\mathrm{SI}} f=0$. We obtain $\Pi\,\Sigma\,\vec{\beta}^*_\mathrm{SI}+\Pi\,\Sigma\,\bar{\Pi}\vec{\beta}=0$. Then:
\begin{equation}\label{eq:optimalchi}
\vec{\beta}^*_\mathrm{SI}=\left ( -\Pi\left(\Pi\,\Sigma\,\Pi+\bar{\Pi}\right)^{-1}\Pi\,\Sigma\,\bar{\Pi} \ \right ) \vec{\beta} ,
\end{equation}
from which the optimal variant of the inequality can be reconstructed: $\vec{\beta}^* = \vec{\beta}_\mathrm{NOS} + \vec{\beta}_\mathrm{SI}^*$, which leads to Eq.~\eqref{eq:optimalvariant}. In all our tests, the matrix $(\Pi\,\Sigma\,\Pi+\bar{\Pi})$ was invertible; when this is not the case, there are subspaces in $V_\mathrm{SI}$ that do not contribute to the variance, and $(\Pi\,\Sigma\,\Pi+\bar{\Pi})^{-1}$ can replaced by the pseudoinverse.

\subsection{Optimality of CHSH for symmetric setups}

\label{app:symarg}
We prove here that in the case of a (2,2,2) setup physicaly symmetric in the two outputs  (such as  the setup used in \cite{Hensen2015} in its noiseless idealization), we have,  for any Bell inequality $\vec{\beta}$ equivalent to CHSH:
\begin{equation}
\vec{\beta}^\top_\mathrm{NOS} \Sigma \vec{\beta}_\mathrm{SI}=0.
\end{equation}
Then, the optimal solution to minimize $\vec{\beta}^{\top}_\mathrm{SI} \Sigma \vec{\beta}_\mathrm{SI}$ is to set $\vec{\beta}^{\top}_\mathrm{SI} = 0$, ie CHSH is the optimal inequality.\\
\\
Let $\rho$ be the symmetry which exchanges $0$ and $1$ on Alice's and Bob's sides. As usual, it can be seen as a relabeling of the outputs, ie a linear map $\vec{\beta}_\mathrm{NOS} \stackrel{\rho}{\mapsto}\vec{\beta}^{\rho}_\mathrm{NOS}$, $\vec{\beta}_\mathrm{SI} \stackrel{\rho}{\mapsto}\vec{\beta}^{\rho}_\mathrm{SI}$. It can also be seen as a transformation of the experimental setup itself: This time, we apply the transformation over the physical setup, transforming the physical state and the measurement operators. This transforms the covarient matrix $\Sigma \stackrel{\rho}{\mapsto}\Sigma^\rho$. Thus, we have $\vec{\beta}^{\rho\top}_\mathrm{NOS} \Sigma \vec{\beta}^\rho_\mathrm{SI}=\vec{\beta}^{\top}_\mathrm{NOS} \Sigma^\rho \vec{\beta}_\mathrm{SI}$.
Moreover, we can show the following:\\
\begin{itemize}
\item[\textit{(i)}] For any Bell inequality $\vec{\beta}$ equivalent to CHSH (i.e. equal over the nosignaling subspace), we have $\vec{\beta}^{\rho}_\mathrm{NOS}=\vec{\beta}_\mathrm{NOS}$ and $\vec{\beta}^{\rho}_\mathrm{SI}=-\vec{\beta}_\mathrm{SI}$.
\begin{proof}
We consider the decomposition given in the \theoremname{} to decompose $\vec{\beta}_\mathrm{NOS}=\vec{\beta}_\mathrm{NO}+\vec{\beta}_\mathrm{marg}+\vec{\beta}_\mathrm{corr}$. As $\vec{\beta}$ is equivalent to $\vec{\beta}^{\mathrm{CHSH}}$, we have $\vec{\beta}_\mathrm{marg}=\vec{\beta}_\mathrm{marg}^{\mathrm{CHSH}}=0$. \sectionname{\ref{sec:decompo}} gives that for any party $M$, $\rho$ is $Id$ over $T_M$, $S^*_M$ and $-Id$ over $\phi_M$. Then, considering Figure~\ref{fig_permutirreps}, we see that $\vec{\beta}_\mathrm{NOS}^\rho=\vec{\beta}_\mathrm{NO}^\rho+\vec{\beta}_\mathrm{corr}^\rho=\vec{\beta}_\mathrm{NO}+\vec{\beta}_\mathrm{corr}=\vec{\beta}_\mathrm{NOS}$ and $\vec{\beta}_\mathrm{SI}^\rho=-\vec{\beta}_\mathrm{SI}$.
\end{proof}

\item[\textit{(ii)}] For any Bell setup symmetric in the two outputs, we have $\Sigma^\rho=\Sigma$.
\begin{proof}
No matter what physical formalism is used to describe the Bell setup (classical, quantum, etc.), the transformation $\rho$ must have a meaning over the physical setup\footnote{For a usual quantum Bell test, it corresponds to a transformation of the setting state and a transformation of the measurement operators. In this case, the probabilities are given by $p_{ab|xy} = {\rm Tr}\left[ (\Pi^A_a \otimes \Pi^B_b) R(\theta^A_x,\theta^B_y) \rho  R(\theta^A_x,\theta^B_y)^\dagger  \right]$ and $\rho$ is a transformation of $\rho$ and the measurement operators $\Pi^M_m$ (with $M=A,B$ and $m=a,b$). For the idealized \cite{Hensen2015}, this expression is unchanged by application of $\rho$.}. The first setup is associated to a covarient matrix $\Sigma$, and the transformed one to $\Sigma^\rho$. If this transformation lets the setup invariant (which is the case in the idealized noise-less setup \cite{Hensen2015}), we must have $\Sigma=\Sigma^\rho$.
\end{proof}
\end{itemize}

\section{Wreath product groups}
\label{app:wreath}

Any given setup of a Bell scenario is represented by a vector $\vec{P}\in V$, which is a representation of a finite group $G$ associated to the scenario. As we will see in the next section, this group is build from an operation over groups, the wreath product $\wr$. Indeed, $G$ is the wreath product of the three groups of permutations of the outcomes, settings and parties: $G=\mathfrak{S}_\mathrm{outcomes}\wr \mathfrak{S}_\mathrm{settings}\wr \mathfrak{S}_\mathrm{parties}$. The vector space $V$ is a representation of this group. More precisely, it is "\textit{the $\mathfrak{S}_{parties}$-primitive representation of the $\mathfrak{S}_{settings}$-imprimitive representation of the $\mathfrak{S}_{outcomes}$-natural representation}".
\\
Here, we will briefly introduce the notions of wreath product and primitive/imprimitive representations in the case of permutation groups. Those structures are used in~\ref{app:step}.

If $\mathfrak{S}_p$ and $\mathfrak{S}^{*}_q$ are two permutation groups, their wreath product is a finite group $G= \mathfrak{S}_p\wr \mathfrak{S}^{*}_q$. Any element $g\in G$ is defined by:
\begin{itemize}
\item an element $\rho\in \mathfrak{S}^{*}_q$
\item for each $0\leq k^*\leq q-1$, an element $\sigma_{k^{*}}\in \mathfrak{S}_p$
\end{itemize}
From now, we will focus on the case $q=3$ and $p=2$ (generalization are straightforward).  Then, $g\equiv (\rho,(\sigma_{0^*},\sigma_{1^*},\sigma_{2^*}))$. 
\begin{table}
\center
\begin{tabular}{l | c | c | c c}
$g_0\in \mathfrak{S}_2\wr \mathfrak{S}^{*}_3$ & 
$\rho$:\begin{tabular}{c} $0^* \rightarrow 1^*$ 
 \\ $1^* \rightarrow 0^*$\\$2^* \rightarrow 2^*$\end{tabular}&
\begin{tabular}{l}
$\sigma_{0^*}$: \begin{tabular}{c}  \footnotesize $0 \rightarrow 1$ \\ \footnotesize $1 \rightarrow 0$ \end{tabular}\\ \hline
$\sigma_{1^*}$: \begin{tabular}{c}  \footnotesize $0 \rightarrow 0$ \\ \footnotesize $1 \rightarrow 1$ \end{tabular}\\ \hline
$\sigma_{2^*}$: \begin{tabular}{c} \footnotesize $0 \rightarrow 1$ \\  \footnotesize $1 \rightarrow 0$ \end{tabular} 
  \end{tabular} \\
\hline 
Transformation of \\ 
\includegraphics{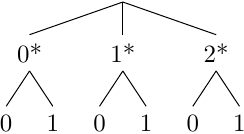}
  &\multicolumn{2}{c}{
\includegraphics{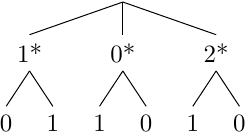}
}
\end{tabular}
\caption{Example of tree representation of an element $g_0$ of $\mathfrak{S}_2\wr \mathfrak{S}^{*}_3$.}
\label{tab:wreath}
\end{table}
The internal product in $G$ may be expressed formally. However, here we introduce it by its most convenient interpretation, where elements of $G$ are seen as transformations of a tree.
We consider a two level tree, the first level having three nodes (labeled $0^*$, $1^*$, $2^*$) and each of these nodes having two leaves  (labeled $0$, $1$).
Then, any $g\in G$ transforms the tree as follows:
\begin{itemize}
\item it permutes the first level with $\rho^*$.
\item for each node number $k^*\in\{0^*,1^*,2^*\}$, it permutes its leaves according to $\sigma_{k^*}$.
\end{itemize}
Then, the product of $g,g'\in G$ is $g''$, the transformation which result in the same tree than the one obtained after successive map of the tree by $g'$ and $g$ and the inverse can be defined in the same way. This gives a group structure.
A simple example $g_0$, with its tree representation, is given in \tablename{~\ref{tab:wreath}}.

\subsection{Primitive and imprimitive representations}

There are two natural ways to create representations of wreath product groups, the primitive and the imprimitive representations. Again, the most convenient is to use the tree representation.

\begin{itemize}

\item For the imprimitive representation, we consider that each path from the root of the tree to a leaf is associated to some vector. Such a path can be labeled with two numbers $j^*$ and $i$, and is associated to a vector $e_{i}\otimes e_{j^*}$. The basis vectors span $V_\mathrm{prim}\cong\mathbb{R}^2\otimes{\mathbb{R}^*}^3$. Then, we apply $g$ to the tree and obtain a new vector. This gives the action of $g$ over any basis vector of $V_\mathrm{prim}$. One can easily find that the basis vector $~e_{1}\otimes e_{0^*}$ is mapped to $~e_{0}\otimes e_{1^*}$ by $g_0$.

\begin{center}
\begin{tabular}{ccc}
\begin{minipage}[c]{4.5cm}
\center
\includegraphics{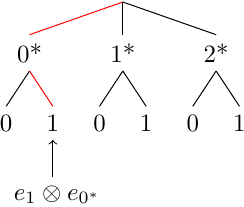}
\end{minipage}
&{\centering $\xrightarrow[]{g}$} &
\begin{minipage}[c]{4.5cm}
\center
\includegraphics{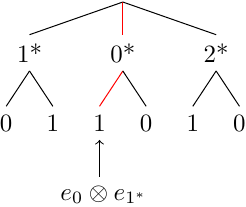}
\end{minipage}
\end{tabular}
\end{center}

Mathematically, an element $(\rho,(\sigma_{0^*},\sigma_{1^*},\sigma_{2^*}))\in \mathfrak{S}_2\wr \mathfrak{S}^{*}_3$ acts on the basis as $e_{i}\otimes e_{j^*}\rightarrow e_{\sigma_{j^*}(i)}\otimes e_{\rho(j^*)}$.

\item For the primitive representation, we consider the eight possible choices of the position of one leaf in each of the three groups of two leaves. Then, such a choice can be labeled with three numbers $i_{0^*}, i_{1^*}, i_{2^*}$, and is associated to a basis vector $e_{i_{0^*}}\otimes e_{i_{1^*}}\otimes e_{i_{2^*}}$. The basis vectors span $V_\mathrm{imprim}\cong\mathbb{R}^2\otimes\mathbb{R}^2\otimes\mathbb{R}^2$. Then, we apply $g$ to the tree and obtain a new vector. This gives the action of g over any basis vector of $V_\mathrm{imprim}$. One can easily find that the basis vector $e_{1}\otimes e_{0}\otimes e_{0}$ is mapped to $e_{0}\otimes e_{0}\otimes e_{1}$ by $g_0$.
\begin{center}
\begin{tabular}{ccc}
\begin{minipage}[c]{5cm}
\center
\includegraphics{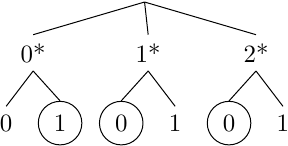}
$e_{1}\otimes e_{0}\otimes e_{0}$
\end{minipage}
&{\centering $\xrightarrow[]{g}$} &
\begin{minipage}[c]{5cm}
\center
\includegraphics{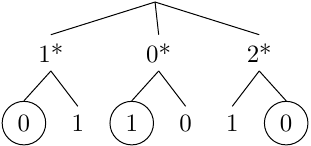}
$e_{0}\otimes e_{0}\otimes e_{1}$
\end{minipage}
\end{tabular}
\end{center}

Mathematically, an element $(\rho,(\sigma_{0^*},\sigma_{1^*},\sigma_{2^*}))\in \mathfrak{S}_2\wr \mathfrak{S}^{*}_3$ acts on the basis as $e_{i_{0^*}}\otimes e_{i_{1^*}}\otimes e_{i_{2^*}}\rightarrow e_{\sigma_{\rho^{-1}(0^*)}(i_{\rho^{-1}(0^*)})}\otimes e_{\sigma_{\rho^{-1}(1^*)}(i_{\rho^{-1}(1^*)})}\otimes e_{\sigma_{\rho^{-1}(2^*)}(i_{\rho^{-1}(2^*)})}$.
\end{itemize}

These two natural representations are not irreducible.

\section*{References}
\bibliographystyle{iopart-num}
\bibliography{biblio}

\end{document}